\documentclass[11pt,english]{article}
\usepackage[T1]{fontenc}
\usepackage[latin9]{inputenc}
\usepackage{geometry}
\usepackage{array}
\usepackage{float}
\usepackage{multirow}
\usepackage{amsmath}
\usepackage{amsthm}
\usepackage{amssymb}
\usepackage{graphicx}
\usepackage{setspace}
\usepackage[authoryear]{natbib}
\makeatletter

\providecommand{\tabularnewline}{\\}

\theoremstyle{definition}
\newtheorem{defn}{\protect\definitionname}
\theoremstyle{plain}
\newtheorem{lem}{\protect\lemmaname}
\theoremstyle{plain}
\newtheorem{thm}{\protect\theoremname}
\theoremstyle{plain}
\newtheorem{cor}{\protect\corollaryname}

\@ifundefined{date}{}{\date{}}

\usepackage{amsmath}

\DeclareMathOperator*{\argmin}{arg\,min}

 \DeclareMathOperator{\supp}{supp}

\usepackage[T1]{fontenc}
\usepackage{lmodern}
\usepackage{comment}
\usepackage[bottom, symbol]{footmisc}

\@ifundefined{showcaptionsetup}{}{%
 \PassOptionsToPackage{caption=false}{subfig}}
\usepackage{subfig}
\makeatother

\usepackage{babel}
\providecommand{\corollaryname}{Corollary}
\providecommand{\definitionname}{Definition}
\providecommand{\lemmaname}{Lemma}
\providecommand{\theoremname}{Theorem}

\begin{document}
\noindent \begin{center}
\textbf{\large{}Censored Quantile Regression with Many Controls}\\
\par\end{center}

\noindent \begin{center}
Seoyun Hong\footnote{Department of Economics, Boston University, seoyun@bu.edu}\\
~\\
March 5, 2023
\par\end{center}
\begin{abstract}
This paper develops estimation and inference methods for censored
quantile regression models with high-dimensional controls. The methods
are based on the application of double/debiased machine learning (DML)
framework to the censored quantile regression estimator of \citet{buchinsky1998alternative}.
I provide valid inference for low-dimensional parameters of interest
in the presence of high-dimensional nuisance parameters when implementing
machine learning estimators. The proposed estimator is shown to be
consistent and asymptotically normal. The performance of the estimator
with high-dimensional controls is illustrated with numerical simulation
and an empirical application that examines the effect of 401(k) eligibility
on savings. 

\medskip

\end{abstract}

\section{Introduction}

Censoring is a common problem in empirical analysis. For example,
income survey data are often right-censored due to top-coding for
high-income workers and left-censored due to the minimum wage or the
wage being naturally bounded by zero. Quantile regression is particularly
effective in analyzing censored outcomes due to the equivariance property
of the quantile function to monotonic transformations, such as censoring.
Using this insight, censored quantile regression estimators stemming
from the pioneer work of \citet{powell1986censored} are widely used
in empirical research as they require weaker distributional assumptions
on the error term compared to the traditional Tobit model. On the
other hand, there is a growing literature on the use of high-dimensional
data for causal inference in econometrics. When relying on the selection-on-observables
assumption for causal inference, researchers are motivated to include
a set of control variables that may be correlated with the main variable
of interest and the outcome. Control variables can often be high-dimensional,
because of the large data sets such as text or scanner data, or from
technical controls created through transformations of raw controls,
including interactions, powers, and b-splines. When the number of
control variables is large relative to the sample size, traditional
econometric methods might not perform well.

I develop estimation and inferential results to a censored quantile
regression model that allows for high-dimensional control variables.
While \citet{buchinsky1998alternative} is one of the most widely
used censored quantile regression estimators with fixed censoring,
it is not applicable when the dimension of control variables $p$
is either comparable to the sample size $N$ or even larger. A common
approach in this setting is to employ regularization, such as by using
the Lasso quantile regression estimator in \citet{belloni2011l1}.
However, this may not be appropriate for a causal inference setting
with observational data. While the researcher includes a large set
of control variables to mitigate omitted variable bias, Lasso selects
variables based on their correlation with the outcome, without considering
their potential confounding effect. Therefore, econometricians have
revisited the classic semiparametric inference problem of making inference
on a low-dimensional parameter in the presence of high-dimensional
nuisance parameters when adopting machine learning methods to traditional
econometric estimators.

The objective of the proposed estimator is to offer a valid inference
procedure for the effect of treatment on the conditional quantile
(i.e., quantile treatment effect) in the presence of high-dimensional
nuisance parameters, such as coefficients on many controls. First,
to estimate the nuisance parameters, I implement machine learning
methods which enable estimation with high-dimensional controls by
performing regularization. However, naively plugging in machine learning
estimates could cause a severe bias in the target parameter due to
the regularization bias and overfitting. Therefore, I employ two tools
introduced in \citet{chernozhukov2018double}: (1) Neyman-orthogonal
score function and (2) cross-fitting algorithm. The Neyman-orthogonal
score ensures that the moment condition is locally insensitive to
the estimates of nuisance parameters. This implies that the moment
condition is less sensitive to regularization mistakes made by the
machine learning estimators. The use of the cross-fitting algorithm
avoids imposing strong restrictions on the growth of entropy and model
complexity. I show that the new estimator is $\sqrt{N}$-consistent
and asyptotically normal under suitable regularity conditions.

I demonstrate the implementation of the estimator by Monte-Carlo simulations
and an empirical application to estimate the quantile treatment effect
of 401(k) eligibility on total financial assets. Simulation results
show that the suggested estimator outperforms both the estimator that
naively plugs in machine learning estimates for nuisance parameters
and the estimator that uses high-dimensional controls without model
selection. Specifically, the confidence intervals based on the developed
estimator yield coverage probabilities close to the nominal level
across different classes of data generating processes, which illustrates
the uniform property. In the empirical application, the developed
estimator produces consistent results when using sets of controls
with different dimensions. This adds credibility to previous research
that has used intuitively chosen low-dimensional control variables
to control for confounding effects. Overall, the proposed estimator
provides a powerful tool for analyzing data with censored outcomes
and high-dimensional controls.

\subsection{Literature review}

There is a large body of literature on censored quantile regression
estimators with fixed censoring. The idea of censored quantile regression
was first introduced by \citet{powell1986censored}, but its implementation
faced computational difficulties due to a non-convex optimization
problem. One approach in the literature aimed to address this issue
by developing the algorithm to implement Powell's estimator; see \citet{fitzenberger199715},
\citet{fitzenberger2007improving}, and \citet{koenker2008censored}.
The other approach proposed alternative estimation methods based on
the result that Powell's estimator is asymptotically equivalent to
running a standard quantile regression on a subsample for which the
conditional quantile is not affected by censoring. A crucial step
in implementing this method is to estimate the unknown censoring probabilities
and select the subsample. To estimate the censoring probability, researchers
have used non-parametric estimation (\citet{buchinsky1998alternative},
\citet{khan2001two}), a probit model (\citet{chernozhukov2002three}),
and a combination of both (\citet{tang2012informative}). \citet{chernozhukov2015quantile}
built upon \citet{chernozhukov2002three} and incorporated endogeneous
regressors. Lastly, \citet{chen2018sequential} leveraged the continuity
of quantile regression coefficients to select the subsample by implementing
sequential estimation over a grid of quantiles. The paper uses the
quantile regression coefficients of the upper quantile to identify
the sub-sample of the lower quantile when the outcome is left-censored.

This work builds on recent advancements in econometrics models for
high-dimensional data and post-selection inference. As noted in the
introduction, the literature addresses the inferential problem for
low-dimensional target parameters in the presence of high-dimensional
nuisance parameters. \citet{belloni2012sparse} presented a linear
instrumental variable model in which the nuisance function was the
optimal instrument estimated using Lasso. \citet{belloni2014inference}
introduced a partially linear model where the nonlinear part of the
model was treated as a nuisance function. \citet{belloni2017program}
developed an estimator for a series of treatment effects, including
average treatment effect and quantile treatment effect, in the presence
of many control variables. \citet{chernozhukov2015valid}, \citet{chernozhukov2018double},
and \citet{chernozhukov2020locally} provided a general framework
of constructing an estimator with a valid post-selection inference
when using machine learning estimators. There have also been works
on the quantile regression with many controls, such as \citet{belloni2011l1}
and \citet{belloni2019valid}, as well as instrument variable quantile
regression by \citet{chernozhukov2018instrumental} and \citet{chen2021debiased}.

Lastly, I want to discuss existing high-dimensional censored quantile
regression models with random censoring. According to \citet{koenker2008censored},
random censoring requires that censoring point is independent of the
outcome conditional on covariates. For example, in a study of the
effect of medical treatment on the time until cancer recurrence, the
outcome may be censored if a patient drops out or has not experienced
recurrence by the end of the study. The assumption of random censoring
is that the censoring point, such as dropout time, is independent
of the outcome, given covariates such as patient's health conditions.
The \citet{zheng2018high} developed a high-dimensional censored quantile
regression model using stochastic integral-based sequential estimation
and penalization, building on the works of \citet{portnoy2003censored}
and \citet{peng2008survival}. \citet{fei2021inference} employed
a splitting and fusing scheme to conduct inferences for all variables
in the model, not just the target parameter. \citet{he2022scalable}
pointed out that the computational algorithms of \citet{zheng2018high}
and \citet{fei2021inference} can be highly inefficient when applied
to high-dimensional data, and proposed a smoothed estimation equation
approach to the sequential estimation procedure. 

This paper makes two main contributions compared to the previous literature.
First, to the best of my knowledge, this is the first work to provide
post-selection inference for censored quantile regression model with
fixed censoring. Aforementioned high-dimensional censored quantile
regression estimators with random censoring can be applied to outcomes
with fixed censoring, but they require more complicated estimation
procedures that are not needed when the censoring is fixed. I argue
that the proposed estimator is more practical in economic contexts
as economists encounter fixed censoring more often. Additionally,
focusing on post-selection inference for target parameters is more
appropriate for causal inference as it accounts for confounding effects
in contrast to naive implementation of machine learning estimators.
Secondly, I derive an orthogonal score function for estimating quantile
treatment effects with censored outcomes, which has advantages even
in low-dimensional setting. In comparison to nonparametric estimation
of the selection probability, the estimator avoids the curse of dimensionality.
In comparison to parametric estimation of the selection probability,
the orthogonal score function allows for model misspecification of
the selection probability. 

\subsection{Plan of the paper}

The rest of the paper is structured as follows. Section 2 reviews
the censored quantile regression estimator of \citet{buchinsky1998alternative}
and highlights its limitations when using machine learning estimators
for estimating nuisance parameters. Section 3 introduces the DMLCQR
estimator, its estimation algorithm, and its asymptotic properties.
Section 4 presents the results of a Monte Carlo simulation, and Section
5 presents the results of an empirical application. Section 6 concludes
the paper, with the proofs of the main results provided in the Appendix.

\textit{Notation. }I work with triangular array data $\left\{ \omega_{i,n};i=1,\dots,n;n=1,2,3,\dots\right\} $
where for each $n$, $\left\{ \omega_{i,n};i=1,\dots,n\right\} $
is defined on the probability space $\left(\Omega,\mathcal{S},P_{n}\right)$.
Each $\omega_{i,n}=(y_{i,n}^{\prime},d_{i,n}^{\prime},z_{i,n}^{\prime})^{\prime}$
is a vector which are independent across $i$ but not necessarily
identically distributed (i.n.i.d.). Therefore, while all parameters
that characterize the distribution of $\left\{ \omega_{i,n};i=1,\dots,n\right\} $
are implicitly indexed by $P_{n}$ and $n$, I omit this dependence
from the notation to maintain simplicity. I use the following notation:
$E_{n}[f]=n^{-1}\sum_{i=1}^{n}f(\omega_{i})$ and $\bar{E}_{n}[f]=E[E_{n}[f]]=\dfrac{1}{n}\sum_{i=1}^{n}E\left[f(\omega_{i})\right]$.
The $l_{2}$-norm is denoted by $\left\Vert \cdot\right\Vert $; the
$l_{0}$-norm $\left\Vert \cdot\right\Vert _{0}$ denotes the number
of nonzero components of a vector; and the $l_{\infty}$-norm $\left\Vert \cdot\right\Vert _{\infty}$
denotes the maximal absolute value in the components of a vector. 

\section{Censored quantile regression estimator}

\subsection{Censored quantile regression}

For a quantile index $\tau\in(0,1)$, consider a partially linear
censored quantile regression model
\begin{align}
\text{y}_{i}^{*} & =d_{i}\theta_{\tau}+g_{\tau}(z_{i})+\epsilon_{i},\;Q_{\tau}(\epsilon_{i}|d_{i},z_{i})=0
\end{align}
where $y_{i}^{*}$ is the latent outcome variable, $d_{i}$ is the
target variable of interest (e.g. a policy or treatment variable),
and the variable $z_{i}$ represents the counfounding factors which
affect the equation through an unknown function $g_{\tau}$. The observed
outcome $y_{i}$ is censored from below such that $y_{i}=\max(y_{i}^{*},c_{i})$.
Note that the censoring point can be different for each observation
and it could be unknown, as the estimator does not require the knowledge
of the censoring point. The $p$-dimensional control variables $x_{i}=X(z_{i})$
are used to approximate the function $g_{\tau}(z_{i})$ which takes
the form
\[
g_{\tau}(z_{i})=x_{i}^{\prime}\beta_{\tau}+r_{\tau i}
\]
where $r_{\tau i}$ is an approximation error. The main parameter
of interest $\theta_{\tau}$ is the quantile treatment effect on the
latent outcome, while $\beta_{\tau}$ and $r_{\tau}$ are nuisance
parameters. 

I introduce the censored quantile regression estimator of \citet{buchinsky1998alternative}.
Let $w_{i}=(d_{i},z_{i})$,$t_{i}=I(y_{i}^{*}>c_{i})$ and $\pi(w_{i})=P[t_{i}=1|w_{i}]$.
In (1), $d_{i}\theta_{\tau}+g_{\tau}(z_{i})$ is the $\tau$th conditional
quantile of $y_{i}^{*}$ given $w_{i}$. Therefore, the conditional
probability that $y_{i}^{*}<d_{i}\theta_{\tau}+g_{\tau}(z_{i})$ given
$w_{i}$, $t_{i}=1$, and $\pi(w_{i})>1-\tau$ is
\[
P(y_{i}^{*}<d_{i}\theta_{\tau}+x_{i}^{\prime}\beta_{\tau}+r_{\tau i}|w_{i},t_{i}=1,\pi(w_{i})>1-\tau)=\dfrac{\pi(w_{i})-(1-\tau)}{\pi(w_{i})}\equiv h_{\tau}(w_{i}).
\]
Then, the population parameter $(\theta_{\tau},\beta_{\tau})$ is
given by
\begin{equation}
(\theta_{\tau},\beta_{\tau})=\argmin_{\theta,\beta}E\left[t_{i}I(h_{\tau i}>0)\rho_{h_{\tau i}}(y_{i}-d_{i}\theta-x_{i}^{\prime}\beta-r_{\tau i})\right]=\argmin_{\theta,\beta}E\left[g(w_{i},\theta,\beta,\pi)\right]
\end{equation}
where $\rho_{\tau}(u)=(\tau-I(u\leq0))u$. Since $\pi(w_{i})$ and
accordingly $h_{\tau}(w_{i})$ are unknown, \citet{buchinsky1998alternative}
use nonparametric estimates $\hat{\pi}(w_{i})$ and $\hat{h}_{\tau}(w_{i})$
from the first step, so the sample estimator is
\begin{equation}
(\hat{\theta}_{\tau},\hat{\beta}_{\tau})=\argmin_{(\theta,\beta)}\dfrac{1}{N}\sum_{i=1}^{N}t_{i}I(\hat{h}_{\tau i}>0)\rho_{\hat{h}_{\tau i}}(y_{i}-d_{i}\theta-x_{i}^{\prime}\beta).
\end{equation}
Note that the estimator in (3) employs a weighted and rotated quantile
regression.

\subsection{High-dimensional setting}

Consider a setting with high-dimensional $x_{i}$, where the dimension
$p$ could be either comparable to the sample size $N$ or larger
than the sample size $(p\gg N)$. The goal is to make inference for
the parameter of interest $\theta_{\tau}$ in the presence of high-dimensional
nuisance parameters $(\pi,\beta_{\tau})$. Conventional methods employed
in the estimator, such as nonparametric estimation and quantile regression,
are not applicable in this setting ($p\gg N$), so machine learning
methods with regularization must be employed. However, when the machine
learning estimates $(\hat{\pi},\hat{\beta}_{\tau})$ are used, the
estimator $\hat{\theta}_{\tau}$ is not necessarily $\sqrt{N}$-consistent.
This is because while using the estimator $(\hat{\pi},\hat{\beta}_{\tau})$
to estimate $\theta_{\tau}$ contributes with a bias of the order
$\left(\left\Vert \hat{\pi}-\pi\right\Vert ,\left\Vert \hat{\beta}_{\tau}-\beta_{\tau}\right\Vert \right)$
in principle, machine learning estimators usually converge slower
than $\sqrt{N}$. Specifically, the score function of (2) for $\theta_{\tau}$
has non-zero pathwise (Gateaux) derivatives with respect to the nuisance
parameter $(\pi,\beta_{\tau})$:
\begin{align*}
\partial_{\pi}E\left[s(w_{i},\theta,\beta,\pi)\right][\pi-\pi_{0}] & =E\left[t_{i}I(h_{\tau i}>0)\dfrac{1-\tau}{\pi^{2}}d_{i}(\pi_{i}-\pi_{i0})\right]\neq0\\
\partial_{\beta}E\left[s(w_{i},\theta,\beta,\pi)\right][\beta-\beta_{0}] & =E\left[t_{i}I(h_{\tau i}>0)f_{i}d_{i}x_{i}^{\prime}(\beta-\beta_{0})\right]\neq0
\end{align*}
where the pathwise derivative is defined in Section 3 and $s(w_{i},\theta,\beta,\pi)=\dfrac{\partial g(w_{i},\theta,\beta,\pi)}{\partial\theta}=t_{i}I(h_{\tau i}>0)(h_{\tau i}-I(y_{i}-d_{i}\theta-x_{i}^{\prime}\beta\leq0))d_{i}$.
This implies that the first-order bias of nuisance parameter estimates
would affect the target parameter, which are regularization and overfitting
bias from using machine learning estimators. Therefore, additional
measures are necessary for valid inference. 

\section{The DMLCQR estimator}

\subsection{The Neyman-orthogonal score}

I refer to the proposed estimator as DMLCQR estimator. The Neyman
orthogonality condition is an important concept in understanding the
estimator, so I introduce the definition in my context following the
definition in \citet{chernozhukov2018double}. Let $(\theta_{0},\eta_{0})$
be the true value of the finite dimensional parameter of interest
$\theta\in\Theta\subset\mathbb{R}^{d_{\theta}}$ and the infinite-dimensional
nuisance parameter $\eta\in\mathcal{T}$, where $\mathcal{T}$ is
a convex subset of some normed vector space. I assume that the moment
conditions $E[\psi(W_{i},\theta_{0},\eta_{0})]=0$ hold. The pathwise
(Gateaux) derivative map $D_{r}:\tilde{\mathcal{T}}\to\mathbb{R}^{d_{\theta}}$
for $\tilde{\mathcal{T}}=\{\eta-\eta_{0}|\eta\in\mathcal{T}\}$ is
defined as

\begin{equation}
D_{r}[\eta-\eta_{0}]\equiv\partial_{r}\{E_{p}[\psi(W,\theta_{0},\eta_{0}+r(\eta-\eta_{0}))]\},\;\eta\in\mathcal{T}
\end{equation}
for all $r\in[0,1)$. For convenience, denote
\begin{equation}
\partial_{\eta}E_{p}[\psi(W,\theta_{0},\eta_{0})][\eta-\eta_{0}]\equiv D_{r}[\eta-\eta_{0}],\;\eta\in\mathcal{T}
\end{equation}
which is the pathwise derivative (4) at $r=0$. Additionally, let
$\mathcal{T}_{N}\subset\mathcal{T}$ be a nuisance realization set
such that estimators of $\eta_{0}$ take values in this set with high
probability. The Neyman orthogonality condition requires that the
derivative in (5) vanishes for all $\eta\in\mathcal{T}_{N}$.
\begin{defn}
\textit{The score function $\psi(W_{i},\theta,\eta)$ obeys the Neyman
orthogonality condition at $(\theta_{0},\eta_{0})$ with respect to
the nuisance parameter realization set $\mathcal{T_{N}\subset\mathcal{T}}$
if $E_{p}[\psi(W_{i},\theta_{0},\eta_{0})]=0$ and the pathwise derivative
map $D_{r}[\eta-\eta_{0}]$ exists for all $r\in[0,1)$ and $\eta\in\mathcal{T_{N}}$,
and vanishes at $r=0$, that is
\[
\partial_{\eta}E_{p}[\psi(W,\theta_{0},\eta_{0})][\eta-\eta_{0}]=0\;\forall\eta\in\mathcal{T}_{N}.
\]
}

I construct a score function that satisfies the Neyman orthogonality
condition in Definition 1. Let $f_{i}=f_{\epsilon_{i}}(0|w_{i})$
denote the conditional density at 0 of the error term $\epsilon_{i}$
in (1). The construction of the orthogonal score function is based
on the linear projection of $d_{i}$ on $x_{i}$, both weighted by
$\sqrt{f_{i}}$ on the subsample that satisfies $t_{i}=1$ and $h_{\tau i}>0$
\begin{equation}
\sqrt{f_{i}}d_{i}=\sqrt{f_{i}}x_{i}^{\prime}\mu_{\tau}+u_{i},\;\bar{E}[t_{i}I(h_{\tau i}>0)\sqrt{f_{i}}x_{i}u_{i}]=0
\end{equation}
where $\mu_{\tau}\in\argmin_{\mu}E\left[t_{i}I(h_{\tau i}>0)f_{i}(d_{i}-x_{i}^{\prime}\mu)^{2}\right]$.
The orthogonal score function i{\small{}s
\begin{align}
\psi(w_{i},\theta,\eta) & =t_{i}I\left(h_{\tau i}>0\right)\left(h_{\tau i}-I(y_{i}-d_{i}\theta_{\tau}-x_{i}'\beta_{\tau}-r_{\tau i}\leq0)\right)(d_{i}-x_{i}^{\prime}\mu_{\tau})+(t_{i}-\pi_{i})I\left(h_{\tau i}>0\right)\dfrac{(1-\tau)}{\pi_{i}}(d_{i}-x_{i}^{\prime}\mu_{\tau})\\
 & =I\left(h_{\tau i}>0\right)\left(t_{i}\left\{ h_{\tau i}-I(y_{i}-d_{i}\theta_{\tau}-x_{i}'\beta_{\tau}-r_{\tau i}\leq0)\right\} +(t_{i}-\pi_{i})\dfrac{(1-\tau)}{\pi_{i}}\right)(d_{i}-x_{i}^{\prime}\mu_{\tau})\nonumber 
\end{align}
w}here $\eta_{\tau}=(\pi,\beta_{\tau},\mu_{\tau})$ is the nuisance
parameter. This leads to a moment condition $E[\psi(w_{i},\theta,\eta_{\tau0})]=0$
to estimate $\theta_{\tau}$ and satisfies the orthogonality condition
defined above:
\begin{align*}
\partial_{\pi}E\left[\psi(w_{i},\theta,\eta)\right][\pi-\pi_{0}] & =0\\
\partial_{\beta_{\tau}}E\left[\psi(w_{i},\theta,\eta)\right][\beta_{\tau}-\beta_{\tau0}] & =0\\
\partial_{\mu_{\tau}}E\left[\psi(w_{i},\theta,\eta)\right][\mu_{\tau}-\mu_{\tau0}] & =0.
\end{align*}
\end{defn}
\begin{lem}
The new score function in (6) obeys the Neyman orthogonality condition.
\end{lem}
The proof of the lemma can be found in the appendix. 

\subsection{Estimation algorithm}

Define the Lasso estimator as
\[
\hat{\theta}\in\argmin_{\theta}E_{n}\left[M\left(y_{i},w_{i},\theta\right)\right]+\dfrac{\lambda}{n}\left\Vert \Gamma\theta\right\Vert _{1}
\]
where $\lambda$ is a penalty level and $\Gamma$ is a diagonal matrix
of penalty loadings. The Post-Lasso estimator is then defined as 

\[
\tilde{\theta}\in\argmin_{\theta}E_{n}\left[M\left(y_{i},w_{i},\theta\right)\right]:\supp(\theta)\subseteq\tilde{T}
\]
where the set $\tilde{T}$ contains $\supp(\hat{\theta})$ and may
also include additional variables considered as important. I will
set $\tilde{T}=\supp(\hat{\theta})$ unless otherwise noted. 

I describe three Lasso regressions to estimate nuisance parameters
$(\hat{\pi},\hat{\beta}_{\tau},\hat{\mu}_{\tau})$. The first is Logit
Lasso regression 
\[
\hat{\alpha}=\argmin_{\alpha}E_{n}\left[-t_{i}\ln\Lambda(w_{i}^{\prime}\alpha)-(1-t_{i})\ln[1-\Lambda(w_{i}^{\prime}\alpha)]\right]+\dfrac{\lambda_{1}}{n}\left\Vert \Gamma_{1}\alpha\right\Vert _{1}
\]
where $\Lambda(u)=\dfrac{\exp(u)}{1+\exp(u)}$. Construct the estimator
of $\pi$ and $h_{\tau}$ by $\hat{\pi}(w_{i})=\Lambda(w_{i}^{\prime}\hat{\alpha})$
and $\hat{h}_{\tau i}=\dfrac{\hat{\pi}_{i}-(1-\tau)}{\hat{\pi}_{i}}$.
The second is Lasso (weighted and rotated) quantile regression
\begin{align}
\left(\hat{\theta}_{\tau},\hat{\beta}_{\tau}\right) & =\argmin_{\theta,\beta}E_{n}\left[t_{i}I\left(\hat{h}_{\tau i}>0\right)\rho_{\hat{h}_{\tau i}}\left(y_{i}-d_{i}\theta-x_{i}^{\prime}\beta\right)\right]+\dfrac{\lambda_{2}}{n}\left\Vert \Gamma{}_{2}(\theta,\beta)\right\Vert _{1}
\end{align}
and the last is Lasso regression with estimated weights
\[
\hat{\mu}_{\tau}=\argmin_{\mu}E_{n}\left[t_{i}I(\hat{h}_{\tau i}>0)\hat{f}_{i}(d_{i}-\mu x_{i})^{2}\right]+\dfrac{\lambda_{3}}{n}\left\Vert \Gamma_{3}\mu\right\Vert _{1}.
\]
The Post-Lasso estimates $(\tilde{\pi},\tilde{\beta}_{\tau},\tilde{\mu}_{\tau})$
could be equivalently used. In the comments below, I discuss the recommended
choices for for $\lambda=(\lambda_{1},\lambda_{2},\lambda_{3})$ and
$\Gamma=(\Gamma_{1},\Gamma_{2},\Gamma_{3})$ and the estimation of
the conditional density function $f_{i}$. I combine the new score
(7) with the cross-fitting algorithm of \citet{chernozhukov2018double}
to propose DMLCQR estimator. The estimator of interest $\tilde{\theta}$
can be constructed either using step 3-1 or 3-2.

\medskip

\hspace{-5mm}\textbf{Algorithm}
\begin{enumerate}
\item Take a $K$-fold random partition $(I_{k})_{k=1}^{K}$ of observation
indices $[N]=\{1,\dots,N\}$. For simplicity, assume that the size
of each fold $I_{k}$ is same with $n=N/K$. For each $k\in[K]=\{1,\dots,K\}$,
define the auxiliary sample $I_{k}^{c}=[N]\setminus I_{k}$.
\item For each $k\in[K]$, construct an ML estimator $\hat{\eta}_{\tau k}=\hat{\eta}_{\tau}((W_{i})_{i\in I_{k}^{c}})$
of $\eta_{\tau0}$ using the auxiliary sample.
\begin{enumerate}
\item Compute $\tilde{\pi}_{k}$ from Logit Post-Lasso regression of $t$
on $d$ and $x$ and calculate $\tilde{h}_{\tau ki}=\dfrac{\tilde{\pi}_{ki}-(1-\tau)}{\tilde{\pi}_{ki}}$. 
\item Compute $(\hat{\theta}_{\tau k},\hat{\beta}_{\tau k})$ from Lasso
(weighted and rotated) quantile regression of $y$ on $d$ and $x$.
Compute the Post-Lasso estimates $(\tilde{\theta}_{\tau k},\tilde{\beta}_{\tau k})$. 
\item Estimate the conditional density $\hat{f}_{k}$.
\item Compute $\tilde{\mu}_{\tau k}$ from the Post-Lasso estimator of $\sqrt{\hat{f}_{k}}d$
on $\sqrt{\hat{f}_{k}}x$ using the subsample where $t_{i}=1$ and
$I(\tilde{h}_{\tau ki}>0)$. Then $\hat{\eta}_{\tau k}=\left(\tilde{\pi}_{k},\tilde{\beta}_{\tau k},\tilde{\mu}_{\tau k}\right)$.
\end{enumerate}
\item[3-1.] Construct the estimator $\check{\theta}_{\tau k}$ as
\[
\check{\theta}_{\tau k}\in\argmin_{\theta}L_{n,k}(\theta)=\dfrac{\left\{ E_{n,k}\left[\psi_{i}(\theta,\hat{\eta}_{\tau k})\right]\right\} ^{2}}{E_{n,k}\left[\psi_{i}(\theta,\hat{\eta}_{\tau k})^{2}\right]}
\]
where $\hat{\psi}_{i}(\theta,\hat{\eta}_{\tau k})=I\left(\tilde{h}_{\tau ki}>0\right)\left(t_{i}\left\{ \tilde{h}_{\tau ki}-I(y_{i}-d_{i}\theta-x_{i}'\tilde{\beta}_{\tau k}\leq0)\right\} +(t_{i}-\tilde{\pi}_{ki})\dfrac{(1-\tau)}{\tilde{\pi}_{ki}}\right)(d_{i}-x_{i}^{\prime}\tilde{\mu}_{\tau k})$
and $E_{n,k}$ is the empirical expectation over the $k$th fold of
the data, $E_{n,k}[\psi(w)]=n^{-1}\sum_{i\in I_{k}}\psi(w_{i})$.
Aggregate the estimators:
\[
\tilde{\theta}_{0}=\dfrac{1}{K}\sum_{k=1}^{K}\check{\theta}_{0,k}.
\]
\item[3-2.] Construct the estimator $\tilde{\theta}_{0}$ as 
\[
\tilde{\theta}_{0}\in\argmin_{\theta}\dfrac{1}{K}\sum_{k=1}^{K}L_{n,k}(\theta).
\]
\end{enumerate}
In Step 2-(b), the Post-Lasso estimator is used as there is a possibility
that the Lasso estimator will not select the variable of interest
$d$ as a relevant control variable. For consistency, the suggested
algorithm uses Post-Lasso estimator for estimating other nuisance
parameters $(\pi_{k},\mu_{k})$, but they can be estimated using the
Lasso estimator. The estimator using Step 3-1 is referred to as DML1
and Step 3-2 as DML2. The difference between the two is the aggregation
method among the estimates in each subsample. Although the theoretical
properties of both methods are the same, \citet{chernozhukov2018double}
note that DML2 might perform better because the pooled objective function
in Step 3-2 is more stable than the separate objective function in
Step 3-1. For the simulation and empirical application results beyond,
I use DML2 estimator.

\subsubsection*{Comment 1. Penalty parameters for Lasso logit and least squares}

I follow \citet{belloni2017program} for setting the penalty level
and estimating the penalty loadings, i.e.,
\begin{align*}
\lambda_{1} & =c\sqrt{n}\Phi^{-1}\left(1-\dfrac{\gamma}{2(p+1)n}\right),\;\lambda_{3}=c\sqrt{n}2\Phi^{-1}\left(1-\dfrac{\gamma}{2p}\right)
\end{align*}
and penalty loading matrices $(\hat{\Gamma}_{1},\hat{\Gamma}_{3})$
can be estimated using the iterative algorithm in the paper.

\subsubsection*{Comment 2. Penalty parameters for Lasso quantile regression}

I adapt the penalty parameters in \citet{belloni2011l1} for weighted
and rotated quatile regression. The penalty parameter used is
\[
\dfrac{\lambda_{2}}{n}=c\times\left\{ (1-\alpha)\text{-quantile of \ensuremath{\left\Vert \Gamma^{-1}E_{n}\left[t_{i}I\left(\hat{h}_{\tau i}>0\right)\left(\hat{h}_{\tau i}-I(U_{i}\leq\hat{h}_{\tau i})\right)x_{i}\right]\right\Vert _{\infty}}}\right\} 
\]
where $U_{1},\cdots,U_{n}$ are i.i.d uniform $(0,1)$ random variables
and $\Gamma_{2}$ is a diagonal matrix with $\Gamma_{2,11}^{2}=E_{n}[d_{i}^{2}]$
and $\Gamma_{2,jj}^{2}=E_{n}[x_{ij}^{2}].$

\subsubsection*{Comment 3. Estimation of conditional density function}

In Step 2-(c), an estimate of the conditional density function $f_{i}$
is required. I follow the estimation method in \citet{belloni2019valid},
in which they use the observation that $\dfrac{1}{f_{i}}=\dfrac{\partial Q(\tau|d_{i},z_{i})}{\partial\tau}$,
where $Q(\cdot|d_{i},z_{i})$ denotes the conditional quantile function
of the latent outcome $y_{i}^{*}$. Let $\hat{Q}(\tau|d_{i},z_{i})$
denote an estimate of the conditional $\tau$-quantile function $Q(\tau|d_{i},z_{i})$
based on Lasso or Post-Lasso estimator of (8) and let $h=h_{n}\to0$
denote a bandwidth parameter. Then an estimator of $f_{i}$ can be
constructed as
\[
\hat{f}_{i}=\dfrac{2h}{\hat{Q}(\tau+h|z_{i},d_{i})-\hat{Q}(\tau-h|z_{i},d_{i})}.
\]

\subsection{Asymptotic properties}

\subsubsection{Regularity conditions}

I present sufficient regularity conditions for the validity of the
main estimation and inference results. The asymptotic results in the
paper builds upon \citet{belloni2019valid}, so the regularity conditions
provided include those in the paper. Specifically, Condiion AS, M,
D are as in \citet{belloni2019valid} and Condition N presents new
regularity conditions not in the paper. Let $c$, $C$, and $q$ be
fixed constants with $c>0$, $C\geq1$, and $q\geq4$, and let $l_{n}\uparrow\infty$,
$\delta_{n}\downarrow0$, and $\Delta_{n}\downarrow0$ be sequences
of positive constants. I assume that the following condition holds
for the data-generating process $P=P_{n}$ for each $n$.

\textit{Condition AS} (1) Let $\{(y_{i},d_{i},x_{i}=X(z_{i}))\}$
be independent random variables satisfying (1) and (6) with $\left\Vert \mu_{0\tau}\right\Vert +\left\Vert \beta_{\tau}\right\Vert +\left|\theta_{\tau}\right|\leq C$.
(2) There exists $s\geq1$ and vectors $\beta_{\tau}$ and $\theta_{\tau}$
such that $x_{i}^{\prime}\theta_{0\tau}=x_{i}^{\prime}\theta_{\tau}+r_{\theta\tau i}$,
$\left\Vert \theta_{\tau}\right\Vert _{0}\leq s$, $\bar{E}[r_{\theta\tau i}^{2}]\leq Cs/n$,
$\left\Vert \theta_{0\tau}-\theta_{\tau}\right\Vert _{1}\leq s\sqrt{\log(pn)/n}$,
and $\left\Vert \beta_{\tau}\right\Vert _{0}\leq s$, $\bar{E}[r_{\tau i}^{2}]\leq Cs/n$.
(3) The conditional distribution function of $\epsilon_{i}$ is absolutely
continuous with continuously differentiable density $f_{\epsilon_{i}|d_{i},z_{i}}(\cdot|d_{i},z_{i})$
such that $0<\underline{f}\leq f_{i}\leq\sup_{t}f_{\epsilon_{i}|d_{i},z_{i}}(t|d_{i},z_{i})\leq\bar{f}\leq C$
and $\sup_{t}f_{\epsilon_{i}|d_{i},z_{i}}^{\prime}(t|d_{i},z_{i})\leq\bar{f^{\prime}}\leq C$. 

\textit{Condition M} (1) We have $\bar{E}[\{(d_{i},x_{i}^{\prime})\xi\}^{2}]\geq c\left\Vert \xi\right\Vert ^{2}$
and $\bar{E}[\{(d_{i},x_{i}^{\prime})\xi\}^{4}]\leq C\left\Vert \xi\right\Vert ^{4}$
for all $\xi\in\mathbb{R}^{p+1}$, 

\hspace{-5mm}$c\leq\min_{1\leq j\leq p}\bar{E}\left[\left|f_{i}x_{ij}v_{i}-E\left[f_{i}x_{ij}v_{i}\right]\right|^{2}\right]^{1/2}\leq\max_{1\leq j\leq p}\bar{E}\left[\left|f_{i}x_{ij}v_{i}\right|^{3}\right]^{1/3}\leq C$.
(2) The approximation error satisfies $\left|\bar{E}\left[f_{i}v_{i}r_{\tau i}\right]\right|\leq\delta_{n}n^{-1/2}$
and $\bar{E}\left[\left(x_{i}^{\prime}\xi\right)^{2}r_{\tau i}^{2}\right]\leq C\left\Vert \xi\right\Vert ^{2}\bar{E}\left[r_{\tau i}^{2}\right]$
for all $\xi\in\mathbb{R}^{p}$. (3) Suppose that $K_{q}=E\left[\max_{1\leq i\leq n}\left\Vert (d_{i},v_{i},x_{i}^{\prime})^{\prime}\right\Vert _{\infty}^{q}\right]^{1/q}$
is finite and satisfies $(K_{q}^{2}s^{2}+s^{3})\log^{3}(pn)\leq n\delta_{n}$
and $K_{q}^{4}s\log(pn)\log^{3}n\leq\delta_{n}n$.

\textit{Condition D}\textbf{ }(1) For $u\in\mathcal{U}$, assume that
$Q_{u}(y_{i}^{*}|z_{i},d_{i})=d_{i}\alpha_{u}+x_{i}^{\prime}\beta_{u}+r_{ui}$,
$f_{ui}=f_{y_{i}|d_{i},z_{i}}(d_{i}\alpha_{u}+x_{i}^{\prime}\beta_{u}+r_{ui}|z_{i},d_{i})\geq c$
where$\bar{E}[r_{ui}^{2}]\leq\delta_{n}n^{-1/2}$ and $\left|r_{ui}\leq\delta_{n}h\right|$
for all $i$ and the vector $\beta_{u}$ satisfies $\left\Vert \beta_{u}\right\Vert _{0}\leq s$.
(2) For $\tilde{s}_{\theta\tau}=s+\frac{ns\log(n\lor p)}{h^{2}\lambda^{2}}+\left(\frac{nh^{\bar{k}}}{\lambda}\right)^{2}$,
suppose $h^{\bar{k}}\sqrt{\tilde{s}_{\theta\tau}\log(pn)}\leq\delta_{n}$,
$h^{-2}K_{q}^{2}s\log(pn)\leq\delta_{n}n$, $\lambda K_{q}^{2}\sqrt{s}\leq\delta_{n}n$,
$h^{-2}s\tilde{s}_{\theta\tau}\log(pn)\leq\delta_{n}n$, $\lambda\sqrt{s\tilde{s}_{\theta\tau}\log(pn)}\leq\delta_{n}n$,
and $K_{q}^{2}\tilde{s}_{\theta\tau}\log^{2}(pn)\log^{3}(n)\leq\delta_{n}n$.

\textit{Condition N} (1) $\pi$ is bounded away from zero, that is
$\pi>c$. (2) Let $\iota_{i}=\pi_{i}-(1-\tau)$. The conditional distribution
function of $\iota_{i}$ is absolutely continuous with continuously
differentiable density $f_{\iota_{i}|d_{i},z_{i}}(\cdot|d_{i},z_{i})$
such that $f_{\iota_{i}}\leq\bar{f}_{\iota}\leq C$.

\subsubsection{Main results}

The following result states that DMLCQR estimator converges to the
true parameter at a $\sqrt{N}$-rate and is approximately normally
distributed. 
\begin{thm}
Suppose that condiitons AS, M, D ,N hold. Then, the estimator satisfies
\[
\bar{\sigma}_{N}^{-1}\sqrt{N}(\tilde{\theta}-\theta)\overset{d}{\to}N(0,1)
\]
where $\bar{\sigma}_{N}^{2}=\left(\bar{E}_{N}[t_{i}I(h_{i}>0)f_{i}d_{i}v_{i}]\right)^{-1}\bar{E}_{N}\left[\psi^{2}(w_{i},\theta_{0},\eta_{0})\right]\left(\bar{E}_{N}[t_{i}I(h_{i}>0)f_{i}d_{i}v_{i}]\right)^{-1}$. 
\end{thm}
The following result establishes that the variance estimator is consistent. 
\begin{thm}
Suppose that condiitons AS, M, D ,N hold. The variance estimator $\hat{\sigma}_{N}^{2}$
is consistent where
\begin{align*}
\hat{\sigma}_{N}^{2} & =\left(\dfrac{1}{K}\sum_{k=1}^{K}E_{n,k}[t_{i}I(\hat{h}_{k,i}>0)\hat{f}_{k,i}d_{i}\hat{v}_{k,i}]\right)^{-1}\dfrac{1}{K}\sum_{k=1}^{K}E_{n,k}\left[\psi^{2}(w_{i},\tilde{\theta},\hat{\eta}_{k})\right]\left(\dfrac{1}{K}\sum_{k=1}^{K}E_{n,k}[t_{i}I(\hat{h}_{k,i}>0)\hat{f}_{k,i}d_{i}\hat{v}_{k,i}]\right)^{-1}.
\end{align*}
\end{thm}
Theorem 1 and 2 can be used for construction of confidence regions,
which are uniformly valid over a large class of data-generating processes.
I formalize the uniform properties as below.
\begin{cor}
Let $\mathcal{P}_{n}$ be the collection of all distributions of $\{(y_{i},d_{i},z_{i}^{\prime})\}_{i=1}^{n}$
for which Conditions AS, M, D, N are satisfied for given $n\geq1$.
Then the confidence interval
\[
CI=\left(\tilde{\theta}\pm\Phi^{-1}(1-\xi/2)\sqrt{\hat{\sigma}_{N}^{2}/N}\right)
\]
obeys
\[
\lim_{N\to\infty}\sup_{P\in\mathcal{P}_{n}}\left|P\left(\theta_{0}\in CI\right)-(1-\xi)\right|=0.
\]
\end{cor}

\section{Monte Carlo simulation}

I conduct a simulation study to evaluate the finite sample performance
of the proposed estimators and confidence regions. I focus on the
case of $\tau=0.5$ and $\tau=0.75$ under the following data-generating
process:
\begin{align}
y_{i}^{*} & =d_{i}\theta_{\tau}+x_{i}^{\prime}(c_{y}\nu_{y})+\varepsilon_{i}\\
d_{i} & =x_{i}^{\prime}(c_{d}\nu_{d})+v_{i}
\end{align}
where $y_{i}=\max(y_{i}^{*},c)$ and $c$ is the 0.3th sample quantile
of $y_{i}^{*}$ in each replication sample. $\theta_{\tau}=1$ is
the parameter of interst, and $x_{i}=(1,z_{i})$ consists of an intercept
and covariate $z_{i}\sim N(0,\Sigma)$. The regressors are correlated
in that $\Sigma_{ij}=\rho^{|i-j|}$ and $\rho=0.5$. The error terms
will follow standard normal distribution, that is $\varepsilon_{i}\sim N(0,1)$
and $v_{i}\sim N(0,1)$. 

For the parameters, I consider the case where
\begin{align*}
\nu_{y} & =(1,1/2,1/3,1/4,1/5,0,0,0,0,0,1,1/2,1/3,1/4,1/5,0,0,\dots,0)\\
\nu_{d} & =(1,1/2,1/3,1/4,1/5,1/6,1/7,1/8,1/9,1/10,0,0,\dots,0).
\end{align*}
Note that a subset of controls have non-zero coefficients indicating
that only a few controls are relevant to the outcome. The coefficients
exhibit a declining pattern which makes $l_{1}$-based model selectors
more likely to make model selection mistakes on variables with smaller
coefficients. The coefficients $(c_{y},c_{d})$ are used to control
$R^{2}$ in equation (9) and (10), which I refer to respectively as
$R_{y}^{2}$ and $R_{d}^{2}$. Changing the values of $(c_{y},c_{d})$
results in varying levels of signal strength, which either makes it
easier or harder for $l_{1}$-based model selectors to detect the
controls with nonzero coefficients. I peform 500 Monte Carlo repetitions
for each design scenario. 

In the first exercise, I focus on the setting with $R_{y}^{2}=R_{d}^{2}=0.75$
by setting the appropriate $(c_{y},c_{d})$ values. I compare the
performance of four different estimators:
\begin{enumerate}
\item Naive post-selection estimator (Naive PS): Estimator of $\theta_{\tau}$
based on post-Lasso estimator of \citet{buchinsky1998alternative}.
Uses model selection but do not use sample splitting and orthogonal
moment.
\item CQR estimator (HDCQR): Estimator of $\theta_{\tau}$ based on \citet{buchinsky1998alternative}
using all $p$ variables.
\item DML-CQR estimator (DMLCQR): The proposed estimator. Uses model selection,
sample splitting, and orthogonal moment.
\item Oracle estimator (Oracle): Estimator of $\theta_{\tau}$ based on
\citet{buchinsky1998alternative} using only relevant controls among
$p$ variables. Assumes the knowledge of the true identity of controls. 
\end{enumerate}
In Table 1 and Figure 1, I provide results for $(n,p)=(500,300)$
and employ 2-fold and 4-fold cross-fitting for the DMLCQR estimator.
As predicted theoretically, the Naive PS and HDCQR estimators show
biases, while the DMLCQR and Oracle estimators exhibit low biases
and their distributions are centered at the true value. This implies
that when incorporating high-dimensional controls in censored quantile
regression models, a naive implementation of model selection leads
to regularization and overfitting biases in the target parameter.
Also, using the original estimator with high-dimensional controls
induces bias and results in a loss of power to learn about the targeted
effect. Therefore, if we use either of these estimators to conduct
hypothesis testing, our inferencial results may be misleading. The
simulation results are consistent across different quantiles.

\newpage

\begin{table}[H]
\textbf{\caption{Summary of simulation results}
}

\medskip
\centering{}%
\begin{tabular}{>{\raggedright}p{1.5cm}>{\centering}p{2cm}>{\centering}p{2cm}>{\centering}p{2cm}>{\centering}p{2cm}>{\centering}p{2cm}>{\centering}p{2cm}}
\hline 
\noalign{\vskip\doublerulesep}
 &  & Naive PS & HDCQR & DMLCQR

(2-fold) & DMLCQR

(4-fold) & Oracle\tabularnewline[\doublerulesep]
\hline 
\noalign{\vskip\doublerulesep}
\multirow{4}{1.5cm}{$\tau=0.5$} & RMSE & 0.227 & 0.175 & 0.131 & 0.130 & 0.068\tabularnewline[\doublerulesep]
\cline{2-7} \cline{3-7} \cline{4-7} \cline{5-7} \cline{6-7} \cline{7-7} 
\noalign{\vskip\doublerulesep}
 & SD & 0.094 & 0.160 & 0.130 & 0.130 & 0.067\tabularnewline[\doublerulesep]
\cline{2-7} \cline{3-7} \cline{4-7} \cline{5-7} \cline{6-7} \cline{7-7} 
\noalign{\vskip\doublerulesep}
 & Bias & 0.206 & -0.070 & 0.016 & 0.013 & -0.010\tabularnewline[\doublerulesep]
\cline{2-7} \cline{3-7} \cline{4-7} \cline{5-7} \cline{6-7} \cline{7-7} 
\noalign{\vskip\doublerulesep}
 & MAE & 0.207 & 0.140 & 0.100 & 0.095 & 0.053\tabularnewline[\doublerulesep]
\hline 
\noalign{\vskip\doublerulesep}
\multirow{4}{1.5cm}{$\tau=0.75$} & RMSE & 0.192 & 0.175 & 0.106 & 0.93 & 0.073\tabularnewline[\doublerulesep]
\cline{2-7} \cline{3-7} \cline{4-7} \cline{5-7} \cline{6-7} \cline{7-7} 
\noalign{\vskip\doublerulesep}
 & SD & 0.086 & 0.161 & 0.101 & 0.091 & 0.073\tabularnewline[\doublerulesep]
\cline{2-7} \cline{3-7} \cline{4-7} \cline{5-7} \cline{6-7} \cline{7-7} 
\noalign{\vskip\doublerulesep}
 & Bias & 0.172 & -0.069 & 0.032 & 0.020 & -0.006\tabularnewline[\doublerulesep]
\cline{2-7} \cline{3-7} \cline{4-7} \cline{5-7} \cline{6-7} \cline{7-7} 
\noalign{\vskip\doublerulesep}
 & MAE & 0.173 & 0.140 & 0.081 & 0.073 & 0.058\tabularnewline[\doublerulesep]
\hline 
\noalign{\vskip\doublerulesep}
\multicolumn{7}{>{\raggedright}p{15cm}}{{\footnotesize{}Note: The table presents root-mean-squared-error (RMSE),
standard deviation (SD), mean bias (Bias), and mean absolute error
(MAE) from the simulation results. Estimators include naive post-selection
estimator (Naive PS),} {\footnotesize{}\citet{buchinsky1998alternative}
with high-dimensional controls (HDCQR), the proposed estimator (DMLCQR)
using 2-fold and 4-fold cross-fitting, and oracle estimator (Oracle)
which applies \citet{buchinsky1998alternative} only with relevant
controls.}}\tabularnewline[\doublerulesep]
\end{tabular}
\end{table}
\begin{figure}[H]
\caption{Histograms of simulation results}

\centering{}\includegraphics[width=0.75\textwidth]{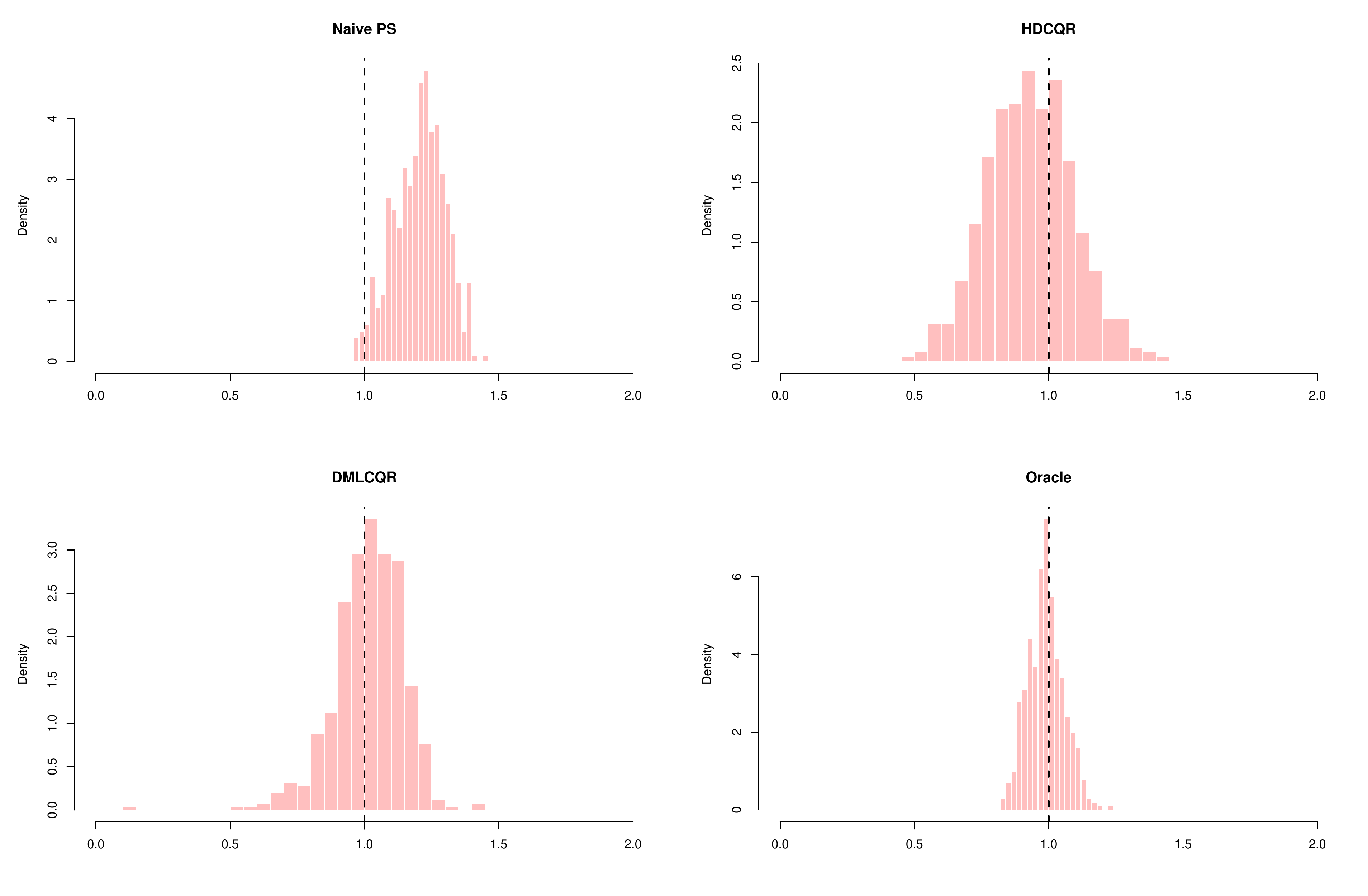}
\end{figure}

\newpage

I now examine confidence intervals of different designs by varying
\[
(R_{d}^{2},R_{y}^{2})\in\left\{ 0.1,0.2,0.3,0.4,0.5,0.6,0.7,0.8,0.9\right\} ^{2},
\]
which gives us 81 different designs. Figure 2 shows the rejection
frequencies of confidence intervals with a nominal significance level
of 5\% for the Naive PS estimator and DMLCQR estimator with 2-fold
cross-fitting. Ideally, the rejection rate should be the nomial level
of 5\% regardless of the underlying data generating process. The confidence
regions for the Naive PS estimator deviate greatly from the ideal
level, while those for the DMLCQR estimator are flat with a rejection
rate close to the nominal level. Therefore, the simulation results
demonstrate that the DMLCQR estimator has satisfactory finite sample
performance compared to conventional alternatives.
\begin{center}
\begin{figure}[H]
\textbf{\caption{Rejection probabilities of confidence regions with nominal coverage
of 95\%}
}

\subfloat{\includegraphics[width=0.5\textwidth]{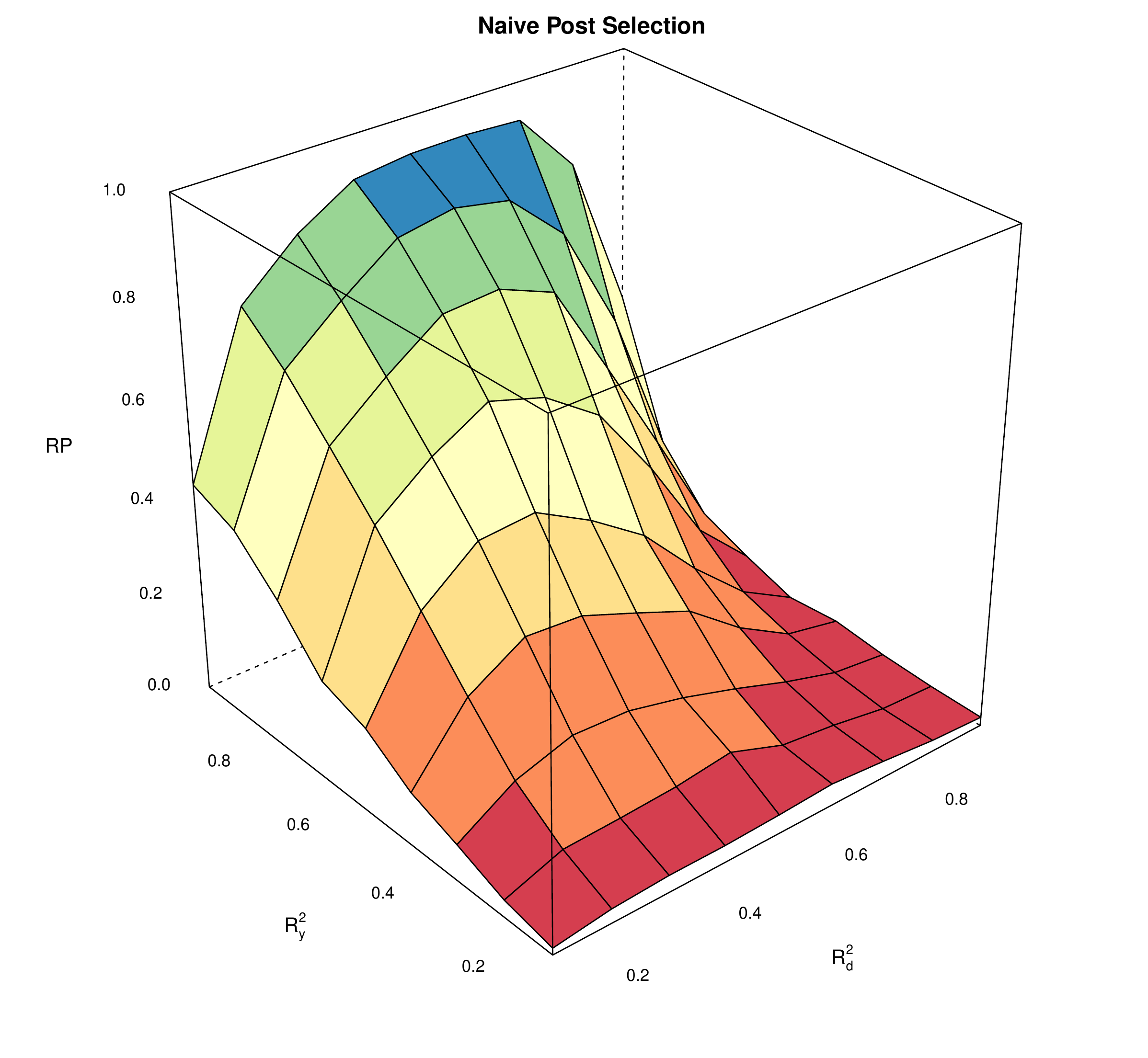}}\subfloat{\includegraphics[width=0.5\textwidth]{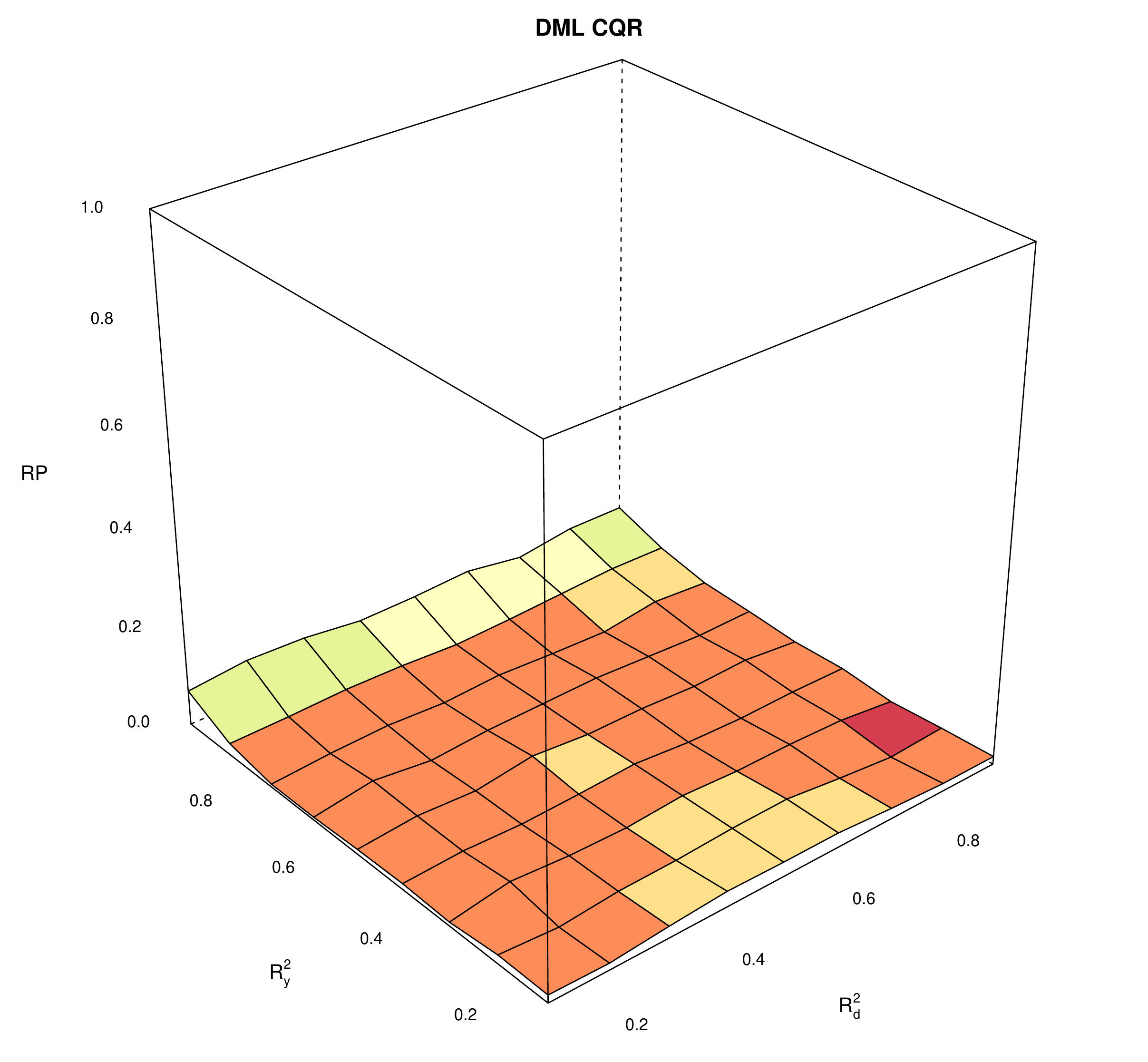}}
\end{figure}
\par\end{center}

\section{Empirical application}

In this section, I apply the DML-CQR estimator to estimate the quantile
treatment effect of 401(k) eligibility on total financial assets.
Since working for a firm that offers a 401(k) plan is non-random,
\citet{poterba1994401} and \citet{poterba1995401} developed a strategy
that takes 401(k) eligibility as exogenous, after conditioning on
control variables including income. They argued that when 401(k) plans
were first introduced, people did not base their employment decisions
on the availability of 401(k) plans, but instead focused on factors
like income and other job characteristics. To validate the selection
on observables assumption, a flexible set of controls that are correlated
with both 401(k) eligibility and the outcome should be included. This
led \citet{belloni2017program} and \citet{chernozhukov2017double}
to use high-dimensional controls and post-selection inference methods
to study the effect of 401(k) eligibility on household's net financial
assets. They revisited previous empirical findings that assumed that
the confounding effects could be adequately controlled for by a small
number of ex ante chosen controls. 

I conduct a similar study to \citet{belloni2017program} and \citet{chernozhukov2017double},
but using total financial asset as the outcome. In the data, 12.4
percent of households reported zero total financial assets, which
is naturally bounded by zero, in contrast to the net financial assets,
which can be negative. Although this is a corner solution outcome
and not from bottom coding, we need to use an econometric model with
a censored outcome to accurately incorporate the data structure. The
treatment of interest is 401(k) eligibility, not participation, so
the parameter of interest is the intent-to-treat effect. \citet{poterba1995401}
conducted a median regression of 401(k) eligibility on total financial
assets, but they did not consider that the outcome was censored and
focused soley on the median quantile. 

I use the same data as \citet{chernozhukov2004effects}. The data
set consists of 9915 household-level observations obtained from the
1991 SIPP. Two households were removed due to reporting negative income,
leaving a sample size of $N=9913$. The outcome variable $Y$ is total
financial assets, which is defined as the sum of IRA balances, 401(k)
balances, and other financial assets. The treatment variable $D$
is a binary indicator of eligibility to enroll in a 401(k) plan. The
vector of raw covariates $X$ includes income, age, family size, years
of education, marriage status indicator, a two-earner status indicator,
a defined benefit pension status indicator, an IRA participation indicator,
and home ownership indicator. Further information about the data can
be found in \citet{chernozhukov2004effects}.

I consider three different sets of controls $f(X)$, modifying the
specification outlined in \citet{belloni2017program}. The first specification
includes indicators of marriage status, two-earner status, defined
benefit pension status, IRA participation, and home ownership status,
as well as second-order polynomials in family size and education,
a third-order polynomial in age, and a quadratic spline in income
with six break points. The second specification adds interactions
between all non-income variables and interactions between non-income
variables and each term in the income spline to the first specification.
The third specification augments interactions between interaction
of non-incomevariables and the income terms. These three specifications
are referred to as the \textit{low-p}, \textit{high-p}, and \textit{very
high-p} specifications, respectively. The dimensions of controls for
the first, second, and third specifications are 20, 182, and 710.
Additionally, I use 20, 157, and 498 variables for methods that do
not use variable selection by removing perfectly collinear terms. 

Figure 3 presents the quantile treatment effect estimates at quantiles
$\tau=0.15,0.2,\dots,0.9$ and 95\% confidence intervals. In the \textit{low-p}
specification, I apply conventional quantile regression and censored
quantile regression estimator of \citet{buchinsky1998alternative}
to compare whether the estimates change according to taking censored
outcome into account in the model. The results indicate that using
censored quantile regression model increases the QTE at lower quantiles
closer to the censoring point. This makes sense as conditional quantiles
at lower quantiles are more likely to be affected by censoring. Therefore,
the result recommends the use of censored outcome models rather than
ignoring the latent data structure. In the \textit{high-p }specification,
I compare the results using DMLCQR estimator and \citet{buchinsky1998alternative},
where the latter could be understood as HDCQR estimator in the simulation.
It shows that incorporating high-dimensional controls in \citet{buchinsky1998alternative}
gives similar results to using quantile regression in the \textit{low-p}
specification, while DMLCQR estimates are consistent with estimates
in the \textit{low-p} specification accounting for censoring. This
implies that implementing the conventional censored quantile regression
estimator with high-dimensional control could lead to somewhat erratic
estimates, close to using quantile regression with low-dimensional
data. Lastly, the application of DMLCQR in the \textit{very high-p}
specification has similar results as in the \textit{high-p} specification.
As the results using DMLCQR with \textit{high-p} and \textit{very
high-p} are very close to censored quantile regression estimates in
the \textit{low-p, }my result complements previous works that were
based on the assumption that a small number of ex ante chosen controls
are enough to control the confounding effect.
\begin{center}
\begin{figure}[H]
\caption{QTE estimates of the effect of 401(k) eligibility on total financial
assets}

\centering{}\includegraphics[width=1\textwidth]{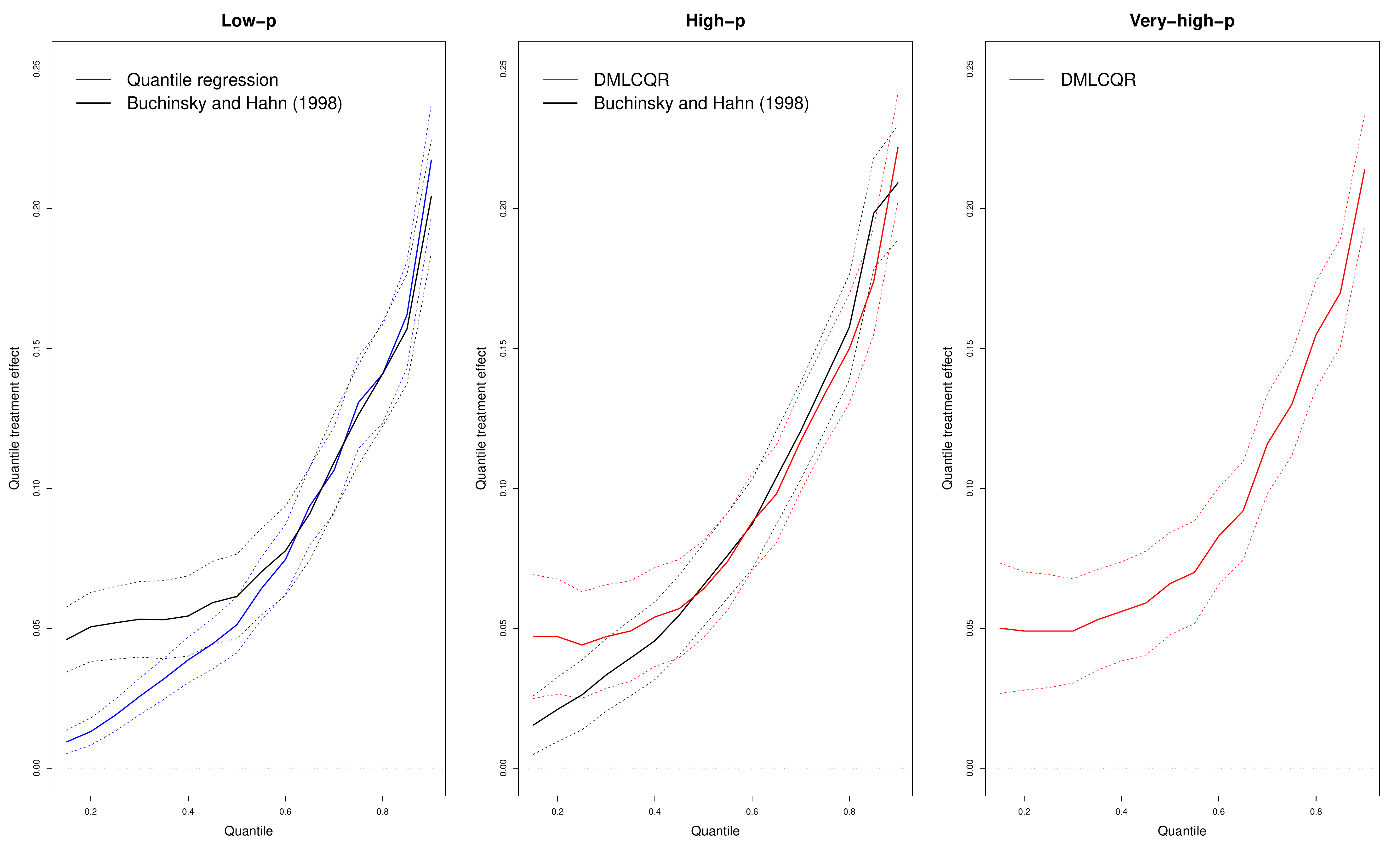}
\end{figure}
\par\end{center}

\section{Conclusion }

Censored quantile regression has been popularly used in the causal
inference literature. Some examples of its applications include the
effect of tax incentives on charitable contributions (\citet{fack2010tax}),
the effect of pollution on productivity (\citet{graff2012impact}),
and the price elasticity of medical expenditure (\citet{kowalski2016censored}).
Empirical researchers frequently face the challenge of selecting important
control variables from a rich set of potential controls, which could
be alleviated by employing machine learning methods with regularization
to estimate high-dimensional nuisance parameters. I propose the DMLCQR
estimator, which improves upon the censored quantile regression estimator
of \citet{buchinsky1998alternative}, and enables researchers to use
high-dimensional control variables while obtaining valid inference
for the target parameter. This additional benefit makes censored quantile
regression more practical for empirical researchers to analyze wider
types of data sets and employ a broader set of machine learning methods
for estimation. 

\newpage

\part*{Appendix}

\section*{Notation}

I work with triangular array data $\left\{ \omega_{i,n};i=1,\dots,n;n=1,2,3,\dots\right\} $
where for each $n$, $\left\{ \omega_{i,n};i=1,\dots,n\right\} $
is defined on the probability space $\left(\Omega,\mathcal{S},P_{n}\right)$.
Each $\omega_{i,n}=(y_{i,n}^{\prime},d_{i,n}^{\prime},z_{i,n}^{\prime})^{\prime}$
is a vector which are independent across $i$ but not necessarily
identically distributed (i.n.i.d.). Therefore, while all parameters
that characterize the distribution of $\left\{ \omega_{i,n};i=1,\dots,n\right\} $
are implicitly indexed by $P_{n}$ and $n$, I omit this dependence
from the notation to maintain simplicity. I use the following notation:
$E_{n}[f]=n^{-1}\sum_{i=1}^{n}f(\omega_{i})$, $\bar{E}_{n}[f]=E[E_{n}[f]]=\dfrac{1}{n}\sum_{i=1}^{n}E\left[f(\omega_{i})\right]$,
and $\mathbb{G}_{n}(f(\omega_{i}))=\dfrac{1}{\sqrt{n}}\sum_{i=1}^{n}\left(f(\omega_{i})-E[f(\omega_{i})]\right)$.
The $l_{2}$-norm is denoted by $\left\Vert \cdot\right\Vert $; the
$l_{0}$-norm $\left\Vert \cdot\right\Vert _{0}$ denotes the number
of nonzero components of a vector; and the $l_{\infty}$-norm $\left\Vert \cdot\right\Vert _{\infty}$
denotes the maximal absolute value in the components of a vector.
For a sequence $\left(t_{i}\right)_{i=1}^{n}$, we denote $\left\Vert t_{i}\right\Vert _{2,n}=\sqrt{E_{n}\left[t_{i}^{2}\right]}$.
We use the notation $a\lesssim b$ to denote $a\leq cb$ for some
constant $c>0$ that does not depend on $n$; and $a\lesssim_{P}b$
to denote $a=O_{p}(b)$. 

\section*{A Proof of the theoretical results}

\subsection*{A.1 Proof of Lemma 1}

{\small{}Moment condition: }From (2), the moment conditions for $\theta_{\tau}$
and $\beta_{\tau}$ are
\begin{align}
E\left[\dfrac{\partial g(w_{i},\theta,\beta,\pi)}{\partial\theta}\right] & =E\left[t_{i}I(h_{\tau i}>0)(h_{\tau i}-I(y_{i}-d_{i}\theta_{\tau}-x_{i}^{\prime}\beta_{\tau}\leq0))d_{i}\right]=0\\
E\left[\dfrac{\partial g(w_{i},\theta,\beta,\pi)}{\partial\beta}\right] & =E\left[t_{i}I(h_{\tau i}>0)(h_{\tau i}-I(y_{i}-d_{i}\theta_{\tau}-x_{i}^{\prime}\beta_{\tau}\leq0))x_{i}\right]=0.
\end{align}
Therefore,
\begin{align*}
E\left[\psi(w_{i},\theta_{0},\eta_{0})\right] & =E\left[I\left(h_{\tau i}>0\right)\left(t_{i}\left\{ h_{\tau i}-I(y_{i}-d_{i}\theta_{\tau}-x_{i}'\beta_{\tau}-r_{\tau i}\leq0)\right\} +(t_{i}-\pi_{i})\dfrac{(1-\tau)}{\pi_{i}}\right)(d_{i}-x_{i}^{\prime}\mu_{\tau})\right]\\
 & =E\left[I\left(h_{\tau i}>0\right)t_{i}\left\{ h_{\tau i}-I(y_{i}-d_{i}\theta_{\tau}-x_{i}'\beta_{\tau}-r_{\tau i}\leq0)\right\} d_{i}\right]\\
 & \;-E\left[I\left(h_{\tau i}>0\right)t_{i}\left\{ h_{\tau i}-I(y_{i}-d_{i}\theta_{\tau}-x_{i}'\beta_{\tau}-r_{\tau i}\leq0)\right\} x_{i}^{\prime}\mu_{\tau}\right]\\
 & \;+E\left[(t_{i}-\pi_{i})I\left(h_{\tau i}>0\right)\dfrac{(1-\tau)}{\pi_{i}}(d_{i}-x_{i}^{\prime}\mu_{\tau})\right]\\
 & =E\left[(t_{i}-\pi_{i})I\left(h_{\tau i}>0\right)\dfrac{(1-\tau)}{\pi_{i}}(d_{i}-x_{i}^{\prime}\mu_{\tau})\right]\;(\because(6),(7))\\
 & =E\left[E[t_{i}-\pi_{i}|w_{i}]I\left(h_{\tau i}>0\right)\dfrac{(1-\tau)}{\pi_{i}}(d_{i}-x_{i}^{\prime}\mu_{\tau})\right]\;(\because\text{Law of iterated expectations})\\
 & =0
\end{align*}
{\small{}and the moment condition} $E\left[\psi(w_{i},\theta_{0},\eta_{0})\right]=0$
holds. 

Neyman orthogonality: The pathwise derivative of (4) in the direction
$\eta-\eta_{0}=(\pi-\pi_{0},\beta-\beta_{\tau},\mu-\mu_{\tau})$ is{\small{}
\begin{align*}
 & \partial_{\pi}E\left[\psi(w_{i},\theta_{0},\eta_{0})\right][\pi-\pi_{0}]\\
 & =E\left[\left\{ t_{i}I\left(h_{i}>0\right)\dfrac{1-\tau}{\pi_{i}^{2}}(d_{i}-x_{i}^{\prime}\mu_{\tau})-t_{i}I\left(h_{i}>0\right)\dfrac{1-\tau}{\pi_{i}^{2}}(d_{i}-x_{i}^{\prime}\mu_{\tau})\right\} \left(\pi_{i}-\pi_{0i}\right)\right]=0
\end{align*}
\begin{align*}
 & \partial_{\beta}E\left[\psi(w_{i},\theta_{0},\eta_{0})\right][\beta-\beta_{\tau0}]\\
 & =\partial_{r}E\left[I\left(h_{\tau i}>0\right)\left(t_{i}\left\{ h_{\tau i}-P(y_{i}-d_{i}\theta_{\tau}-x_{i}'\beta_{\tau}-r_{\tau i}\leq x_{i}^{\prime}(\tilde{\beta}-\beta_{\tau})|w_{i})\right\} +(t_{i}-\pi_{i})\dfrac{(1-\tau)}{\pi_{i}}\right)(d_{i}-x_{i}^{\prime}\mu_{\tau})\right]\Big|_{r=0}\\
 & =-E\left[I\left(h_{\tau i}>0\right)t_{i}f(x_{i}^{\prime}(\tilde{\beta}-\beta_{\tau})|w_{i})(d_{i}-x_{i}^{\prime}\mu_{\tau})x_{i}^{\prime}\left(\beta-\beta_{\tau}\right)\right]\Big|_{r=0}\\
 & =-E\left[t_{i}I\left(h_{i}>0\right)f_{i}(d_{i}-x_{i}^{\prime}\mu_{\tau})x_{i}^{\prime}\left(\beta-\beta_{\tau}\right)\right]=0\;(\because\text{Definition of \ensuremath{\mu_{\tau}}})\\
 & =0
\end{align*}
\begin{align*}
 & \partial_{\mu}E\left[\psi(w_{i},\theta_{0},\eta_{0})\right][\mu-\mu_{\tau0}]\\
 & =-E\left[t_{i}I\left(h_{i}>0\right)\left(h_{i}-I\left\{ y_{i}\leq d_{i}\theta_{\tau}+x_{i}^{\prime}\beta_{\tau}+r_{\tau i}\right\} \right)x_{i}^{\prime}\left(\mu-\mu_{\tau}\right)\right]-E\left[(t_{i}-\pi_{i})I\left(h_{i}>0\right)\dfrac{(1-\tau)}{\pi_{i}}x_{i}^{\prime}\left(\mu-\mu_{\tau}\right)\right]\\
 & =-E\left[t_{i}I\left(h_{i}>0\right)\left(h_{i}-F_{y}\left\{ d_{i}\theta_{\tau}+x_{i}^{\prime}\beta_{\tau}+r_{\tau i}|w_{i},t_{i}=1,\pi_{i}>1-\tau\right\} \right)x_{i}^{\prime}\tilde{\mu}\right]\\
 & \;\;-E\left[E[t_{i}-\pi_{i}|w_{i}]\cdot I\left(h_{i}>0\right)\dfrac{(1-\tau)}{\pi_{i}}x_{i}^{\prime}\tilde{\mu}\right]=0.
\end{align*}
}{\small\par}

\subsection*{A.2 Assumption IQR}

I suppress $\tau$ notation in the parameters for the sake of simplicity.

\subsubsection*{A.2.1 Assumption}

For $k\in[K]$, I assume that
\[
L_{n}(\check{\theta}_{k})\leq\min_{\theta\in\Theta_{\tau}}L_{n}(\theta)+\delta_{n}n^{-1}.
\]
I suppress $k$ notation in the parameters. We define
\begin{align*}
\Gamma(\tilde{\theta},\tilde{\eta}) & =\bar{E}\left[\psi_{\theta,\eta}(y_{i},d_{i},x_{i})\right]|_{\theta=\tilde{\theta},h=\tilde{\eta}}\\
\Gamma_{\eta}(\tilde{\theta},\tilde{\eta})[\hat{\eta}-\eta_{0}] & =\lim_{t\to0}\dfrac{\Gamma(\tilde{\theta},\tilde{\eta}+t[\hat{\eta}-\eta_{0}])-\Gamma(\tilde{\theta},\tilde{\eta})}{t}
\end{align*}
where $\Gamma_{\eta}$ is the directional derivative with respect
to $[\hat{\eta}-\eta_{0}]$ at $(\tilde{\theta},\tilde{\eta})$.

We denote the nuisance parameters as $\eta=(g,v,\pi)$ where $v_{i}=d_{i}-x_{i}^{\prime}\mu$.
We define the score with $(\theta,\eta)$ as {\small{}
\begin{align}
\psi_{\theta,\eta}(w_{i}) & =I\left(h_{i}>0\right)\left(t_{i}\left\{ h_{i}-I(y_{i}-d_{i}\theta-g_{i}\leq0)\right\} +(t_{i}-\pi_{i})\dfrac{(1-\tau)}{\pi_{i}}\right)v_{i}.
\end{align}
}For some sequences $\delta_{n}\to0$ and $\Delta_{n}\to0$, we let
$\overline{\mathcal{F}}$ denote a set of functions such that each
element $\tilde{\eta}=(\tilde{g},\tilde{v},\tilde{\pi})\in\overline{\mathcal{F}}$
satisfies
\begin{align}
\bar{E}\left[(\tilde{g}_{i}-g_{i})^{2}\right] & \lesssim\delta_{n}n^{-1/2}\\
\left|\bar{E}[f_{i}v_{i}\{\tilde{g}_{i}-g_{i}\}]\right| & \leq\delta_{n}n^{-1/2}\\
\bar{E}[(\tilde{v}_{i}-v_{i})^{2}] & \leq\delta_{n}\\
\bar{E}[\left|\tilde{g}_{i}-g_{i}\right|\left|\tilde{v}_{i}-v_{i}\right|] & \leq\delta_{n}n^{-1/2}\\
\bar{E}\left[\left|v_{i}\right|(\tilde{g}_{i}-g_{i})^{2}\right] & \leq\delta_{n}n^{-1/2}\\
\bar{E}\left[\left|\tilde{v}_{i}-v_{i}\right|(\tilde{g}_{i}-g_{i})^{2}\right] & \leq\delta_{n}n^{-1/2}
\end{align}
and with probability $1-\Delta_{n}$ we have
\begin{equation}
\sup_{|\theta-\theta_{0}|\leq\delta_{n}^{2},\tilde{\eta}\in\mathcal{\overline{\mathcal{F}}}}\left|(E_{n}-\bar{E})[\psi_{\theta,\tilde{\eta}}-\psi_{\theta,\eta_{0}}]\right|\leq\delta_{n}n^{-1/2}.
\end{equation}
We assume the estimated functions $\hat{\eta}=(\hat{g},\hat{v},\hat{\pi})$
satisfy the following condition.

\medskip

\textbf{Condition IQR. }Let $\{(y_{i},d_{i},z_{i})\}$ be random variables
independent across $i$ satisfying (1) and (5). Suppose there are
positive constants $0<c\leq C<\infty$ such that: 

\hspace{-5mm}(1) $f_{y_{i}|d_{i},z_{i}}(y|d_{i},z_{i})\leq\bar{f}$,
$f_{y_{i}|d_{i},z_{i}}^{\prime}(y|d_{i},z_{i})\leq\bar{f^{\prime}}$;
$\bar{E}[v_{i}^{4}]+\bar{E}[d_{i}^{4}]\leq C$; $c\leq|\bar{E}[t_{i}I\left(h_{i}>0\right)f_{i}d_{i}v_{i}]|$

\hspace{-5mm}(2) $\left\{ \theta:|\theta-\theta_{0}|\leq n^{-1/2}/\delta_{n}\right\} \subset\Theta_{\tau}$
where $\Theta_{\tau}$ is a (possibly random) compact interval.

\hspace{-5mm}(3) with probability at least $1-\Delta_{n}$ the estimated
functions $\hat{\eta}=(\hat{g},\hat{v},\hat{\pi})\in\overline{\mathcal{F}}$
and
\[
|\check{\theta}-\theta_{0}|\leq\delta_{n}\;\text{and }E_{n}[\psi_{\check{\theta},\hat{\eta}}(y_{i},d_{i},z_{i})]\leq\delta_{n}n^{-1/2}.
\]

\hspace{-5mm}(4) with probability at least $1-\Delta_{n}$ the estimated
functions $\hat{\eta}=(\hat{g},\hat{v},\hat{\pi})$ satisfy
\begin{align}
\left\Vert \hat{v}_{i}-v_{i}\right\Vert _{2,n} & \leq\delta_{n}\\
\left\Vert \hat{\pi}_{i}-\pi_{i}\right\Vert _{2,n} & \leq\delta_{n}\\
\left\Vert I\left(\hat{h}_{i}>0\right)-I\left(h_{i}>0\right)\right\Vert _{2,n} & \leq\delta_{n}\\
\left\Vert I\left(\left|\epsilon_{i}\right|\leq\left|d_{i}(\theta-\check{\theta})+g_{i}-\hat{g}_{i}\right|\right)\right\Vert _{2,n} & \leq\delta_{n}^{2}.
\end{align}

\subsubsection*{A.2.2 Proof}

\textbf{Condition IQR (1)} requires conditions on the probability
density function that are assumed in Condition AS (3). $\bar{E}[v_{i}^{4}]+\bar{E}[d_{i}^{4}]\leq C$
is implied by Condition M (1) using $\xi_{1}=(1,0^{\prime})^{\prime}$
and $\xi_{2}=(1,-\mu_{\tau}^{\prime})^{\prime}$ since $\bar{E}[d_{i}^{4}]=\bar{E}[\{(d_{i},x_{i}^{\prime})\xi_{1}\}^{4}]\leq C$
and $\bar{E}[v_{i}^{4}]\leq\bar{E}[\{(d_{i},x_{i}^{\prime})\xi_{2}\}^{4}]\leq C(1+\left\Vert \mu_{\tau}\right\Vert )^{4}$
and $\left\Vert \mu_{\tau}\right\Vert \leq C$ as assumed in Condition
AS (1). $c\leq|\bar{E}[t_{i}I\left(h_{i}>0\right)f_{i}d_{i}v_{i}]|$
is implied by $\bar{E}[t_{i}I\left(h_{i}>0\right)f_{i}d_{i}v_{i}]=\bar{E}[t_{i}I\left(h_{i}>0\right)f_{i}v_{i}^{2}]\geq\underline{f}$.

\textbf{Condition IQR (3)}: I first verify $\hat{\eta}=(\hat{g},\hat{v},\hat{\pi})\in\overline{\mathcal{F}}$.
\citet{belloni2019valid} show $\bar{E}\left[\left\{ x_{i}^{\prime}(\tilde{\beta}-\beta)\right\} ^{2}\right]\lesssim_{P}\sqrt{\dfrac{s\log(pn)}{n}}$
and $\bar{E}[\left\{ x_{i}^{\prime}(\tilde{\mu}-\mu)\right\} ^{2}]\lesssim\dfrac{1}{h^{2}}\dfrac{s\log(pn)}{n}+h^{2\bar{k}}+\dfrac{\lambda^{2}}{n^{2}}s$where
$h$ is the bandwidth for estimating $\hat{f}_{i}$. I verify equation
(9) to (14) under conditions.

\[
\bar{E}\left[(\tilde{g}_{i}-g_{i})^{2}\right]\leq\bar{E}\left[\left\{ x_{i}^{\prime}(\tilde{\beta}-\beta)\right\} ^{2}\right]+\bar{E}\left[r_{i}^{2}\right]\lesssim\dfrac{s\log(pn)}{n}\lesssim\delta_{n}n^{-1/2}\;(\because\text{Condition AS (2)})
\]
\begin{align*}
\left|\bar{E}[t_{i}I\left(h_{i}>0\right)f_{i}v_{i}(\tilde{g}_{i}-g_{i})]\right| & \leq\left|\bar{E}[t_{i}I\left(h_{i}>0\right)f_{i}v_{i}x_{i}^{\prime}(\tilde{\beta}-\beta)]\right|+\left|\bar{E}[t_{i}I\left(h_{i}>0\right)f_{i}v_{i}r_{i}]\right|\\
 & =\left|\bar{E}[t_{i}I\left(h_{i}>0\right)\sqrt{f}_{i}u_{i}x_{i}^{\prime}(\tilde{\beta}-\beta)]\right|+\left|\bar{E}[t_{i}I\left(h_{i}>0\right)f_{i}v_{i}r_{i}]\right|\\
 & =\left|\bar{E}[t_{i}I\left(h_{i}>0\right)f_{i}v_{i}r_{i}]\right|\;(\because\bar{E}[t_{i}I(h_{\tau i}>0)\sqrt{f_{i}}x_{i}u_{i}]=0)\\
 & \lesssim\delta_{n}n^{-1/2}\;(\because\text{Condition M (2)})
\end{align*}
\[
\bar{E}[(\tilde{v}_{i}-v_{i})^{2}]=\bar{E}[\left\{ x_{i}^{\prime}(\tilde{\mu}-\mu)\right\} ^{2}]\lesssim\dfrac{1}{h^{2}}\dfrac{s\log(pn)}{n}+h^{2\bar{k}}+\dfrac{\lambda^{2}}{n^{2}}s\lesssim\delta_{n}
\]
under $h^{-2}s\log(pn)\leq\delta_{n}n$, $h\leq\delta_{n}$, and $\lambda^{2}s\leq\delta_{n}n^{2}$.

\begin{align*}
\bar{E}[\left|\tilde{g}_{i}-g_{i}\right|\left|\tilde{v}_{i}-v_{i}\right|] & =\bar{E}[\left|\tilde{g}_{i}-g_{i}\right|\left|x_{i}^{\prime}(\tilde{\mu}-\mu)\right|]\\
 & \leq\bar{E}[(\tilde{g}_{i}-g_{i})^{2}]^{1/2}\bar{E}[\left\{ x_{i}^{\prime}(\tilde{\mu}-\mu)\right\} ^{2}]^{1/2}\\
 & \lesssim\left\{ s\log(pn)/n\right\} ^{1/2}\left\{ \dfrac{1}{h}\sqrt{\dfrac{s\log(pn)}{n}}+h^{\bar{k}}+\lambda\dfrac{\sqrt{s}}{n}\right\} \lesssim\delta_{n}n^{-1/2}
\end{align*}
under $h^{-2}s^{2}\log^{2}(pn)\leq\delta_{n}^{2}n$ and $h^{\bar{k}}\sqrt{s\log(pn)}\leq\delta_{n}$.

\begin{align*}
\bar{E}\left[\left|v_{i}\right|(\tilde{g}_{i}-g_{i})^{2}\right] & =\bar{E}\left[\left|v_{i}\right|\left\{ x_{i}^{\prime}(\tilde{\beta}-\beta)\right\} ^{2}\right]+\bar{E}\left[\left|v_{i}\right|r_{i}^{2}\right]\\
 & \leq\left(\bar{E}\left[v_{i}^{2}\left\{ x_{i}^{\prime}(\tilde{\beta}-\beta)\right\} ^{2}\right]\bar{E}\left[\left\{ x_{i}^{\prime}(\tilde{\beta}-\beta)\right\} ^{2}\right]\right)^{1/2}+\left(\bar{E}\left[v_{i}^{2}r_{i}^{2}\right]\bar{E}\left[r_{i}^{2}\right]\right)^{1/2}\\
 & \lesssim\left\{ s\log(pn)/n\right\} ^{1/2}\left\{ \bar{E}\left[v_{i}^{2}\left\{ x_{i}^{\prime}(\tilde{\beta}-\beta)\right\} ^{2}\right]+\bar{E}\left[v_{i}^{2}r_{i}^{2}\right]\right\} ^{1/2}\\
 & \lesssim s\log(pn)/n\lesssim\delta_{n}n^{-1/2}
\end{align*}
since $\left\Vert (1,-\mu)\right\Vert \leq C$, $\bar{E}[\{(d_{i},x_{i}^{\prime})\xi\}^{4}]\leq C\left\Vert \xi\right\Vert ^{4}$,
and $\bar{E}\left[\left(x_{i}^{\prime}\xi\right)^{2}r_{\tau i}^{2}\right]\leq C\left\Vert \xi\right\Vert ^{2}\bar{E}\left[r_{\tau i}^{2}\right]$.
Lastly,

\begin{align*}
\bar{E}\left[\left|\tilde{v}_{i}-v_{i}\right|(\tilde{g}_{i}-g_{i})^{2}\right] & =\bar{E}\left[\left|x_{i}^{\prime}(\tilde{\mu}-\mu)\right|(\tilde{g}_{i}-g_{i})^{2}\right]\\
 & =\bar{E}\left[\left|x_{i}^{\prime}(\tilde{\mu}-\mu)\right|(\tilde{g}_{i}-g_{i})^{2}\right]\\
 & \lesssim s\log(pn)/n\lesssim\delta_{n}n^{-1/2}.
\end{align*}

Next, I verify equation (20). Let $\mathcal{F}_{1,\eta}=\left\{ \psi_{\theta,\eta}:\left|\theta-\theta_{0}\right|\leq\delta_{n}^{2}\right\} $
and $\mathcal{F}_{2}=\left\{ \psi_{\theta,\tilde{\eta}}-\psi_{\theta,\eta_{0}}:\left|\theta-\theta_{0}\right|\leq\delta_{n}^{2}\right\} $.
Note that $\mathcal{F}_{2}\subset\mathcal{F}_{1,\tilde{\eta}}-\mathcal{F}_{1,\eta}$.
Since $\mathcal{F}_{1,\eta}$ is the produce of a VC class of dimension
1 with the random variable $v$, its entropy nuber satisfies $\log N\left(\epsilon\left\Vert F_{1,\eta}\right\Vert _{Q,2},\mathcal{F}_{1,\eta},\left\Vert \cdot\right\Vert _{Q,2}\right)\leq v\log(a/\epsilon)$
where $F_{1,\eta}$ is a measurable envelope for $\mathcal{F}_{1,\eta}$
that satisfies $\left\Vert F_{1,\eta}\right\Vert _{P,q}\leq c_{1}$.
Therefore, $\left\Vert F_{2}\right\Vert _{P,q}=\left\Vert F_{1,\eta}\right\Vert _{P,q}+\left\Vert F_{1,\tilde{\eta}}\right\Vert _{P,q}\leq2c_{1}$
by the triangle inequality and $\log N\left(\epsilon\left\Vert F_{2}\right\Vert _{Q,2},\mathcal{F}_{2},\left\Vert \cdot\right\Vert _{Q,2}\right)\leq2v\log(2a/\epsilon)$
by the proof of Theorem 3 in \citet{andrews1994empirical}. To bound
$\sup_{f\in\mathcal{F}_{2}}\left|\mathbb{G}_{n}(f)\right|$, I apply
Theorem 5.1 of \citet{chernozhukov2014gaussian} conditional on $(w_{i})_{i\in I^{c}}$
so that $\tilde{\eta}$ can be treated as fixed. By the proof of Theorem
2 in Section A.5, note that for $f_{2}\in\mathcal{F}_{2}$,
\begin{align*}
 & \left\Vert f_{2}\right\Vert _{P,2}^{2}\\
 & =\bar{E}\left[\left\{ \psi(w_{i},\theta,\tilde{\eta})-\psi(w_{i},\theta,\eta_{0})\right\} ^{2}\right]\\
 & \lesssim\bar{E}\left[(\hat{v}_{i}-v_{i})^{2}\right]+\bar{E}\left[(\hat{\pi}_{i}-\pi_{i})^{2}\right]+\bar{E}\left[\left(I\left(\hat{h}_{i}>0\right)-I\left(h_{i}>0\right)\right)^{2}\right]+\bar{E}\left[v_{i}^{2}I\left(\left|\epsilon_{i}\right|\leq\left|d_{i}(\theta-\tilde{\theta})+g_{i}-\tilde{g}_{i}\right|\right)^{2}\right]\\
 & \lesssim\dfrac{1}{h^{2}}\dfrac{s\log(pn)}{n}+h^{2\bar{k}}+\dfrac{\lambda^{2}}{n^{2}}s.
\end{align*}
Applying the theorem with with $\sigma=C\left\{ \dfrac{1}{h}\sqrt{\dfrac{s\log(pn)}{n}}+h^{\bar{k}}+\dfrac{\lambda}{n}\sqrt{s}\right\} $,$a=pn$,
$v=s$, $\left\Vert M\right\Vert _{P,2}\leq C$ yields
\begin{align*}
\sup_{|\theta-\theta_{0}|\leq\delta_{n}^{2},\tilde{\eta}\in\mathcal{\overline{\mathcal{F}}}}\left|(E_{n}-\bar{E})[\psi_{\theta,\tilde{\eta}}-\psi_{\theta,\eta_{0}}]\right| & \lesssim\dfrac{1}{\sqrt{n}}\left\{ \left(\dfrac{1}{h}\sqrt{\dfrac{s\log(pn)}{n}}+h^{\bar{k}}+\dfrac{\lambda}{n}\sqrt{s}\right)\sqrt{s\log(pn)}+\dfrac{1}{\sqrt{n}}sC\log(pn)\right\} \\
 & \lesssim\delta_{n}n^{-1/2}.
\end{align*}

\textbf{Condition IQR (4)}: I verify equations (16) to (19).

\[
\left\Vert \hat{v}_{i}-v_{i}\right\Vert _{2,n}\lesssim_{P}\bar{E}\left[(\hat{v}_{i}-v_{i})^{2}\right]^{1/2}\lesssim\delta_{n}^{1/2}
\]
\[
\left\Vert \hat{\pi}_{i}-\pi_{i}\right\Vert _{2,n}\lesssim_{P}\bar{E}\left[(\hat{\pi}_{i}-\pi_{i})^{2}\right]^{1/2}\lesssim\sqrt{\dfrac{s\ln(pn)}{n}}=\delta_{n}
\]
\begin{align*}
\left\Vert I\left(\hat{h}_{i}>0\right)-I\left(h_{i}>0\right)\right\Vert _{2,n} & =\left\Vert I\left(\hat{\pi}_{i}>1-\tau\right)-I\left(\pi_{i}>1-\tau\right)\right\Vert _{2,n}\\
 & =\sqrt{\dfrac{1}{n}\sum_{i=1}^{n}\left\{ I\left(\hat{\pi}_{i}>1-\tau\right)-I\left(\pi_{i}>1-\tau\right)\right\} ^{2}}
\end{align*}
Note that 
\begin{align*}
\left|I\left(\hat{\pi}_{i}>1-\tau\right)-I\left(\pi_{i}>1-\tau\right)\right| & =\left|I\left(\pi_{i}-\varepsilon_{i}>1-\tau\right)-I\left(\pi_{i}>1-\tau\right)\right|\;(\because\varepsilon_{i}=\hat{\pi}_{i}-\pi_{i})\\
 & =\left|I(\iota_{i}>\varepsilon_{i})-I(\iota_{i}>0)\right|\;(\because\iota_{i}=\pi_{i}-(1-\tau))\\
 & \leq I(\left|\iota_{i}\right|<\left|\varepsilon_{i}\right|).
\end{align*}
Therefore,
\begin{align*}
\left\Vert I\left(\hat{\pi}_{i}>1-\tau\right)-I\left(\pi_{i}>1-\tau\right)\right\Vert _{2,n}^{2} & \leq\dfrac{1}{n}\sum_{i=1}^{n}I(\left|\iota_{i}\right|<\left|\hat{\pi}_{i}-\pi_{i}\right|)\\
 & \lesssim_{P}\bar{E}\left[I(\left|\iota_{i}\right|<\left|\hat{\pi}_{i}-\pi_{i}\right|)\right]\\
 & \leq\bar{f}_{\iota}\dfrac{1}{n}\sum_{i=1}^{n}\left|\hat{\pi}_{i}-\pi_{i}\right|\lesssim\sqrt{\dfrac{1}{n}\sum_{i=1}^{n}\left(\hat{\pi}_{i}-\pi_{i}\right)^{2}}\lesssim_{P}\sqrt{\dfrac{s\ln(pn)}{n}}=\delta_{n}.
\end{align*}
\begin{align*}
\left\Vert I\left(\left|\epsilon_{i}\right|\leq\left|d_{i}(\theta-\check{\theta})+g_{i}-\tilde{g}_{i}\right|\right)\right\Vert _{2,n}^{2} & =E_{n}\left[I\left(\left|\epsilon_{i}\right|\leq\left|d_{i}(\theta-\tilde{\theta})\right|\right)\right]+E_{n}\left[I\left(\left|\epsilon_{i}\right|\leq\left|x_{i}^{\prime}(\tilde{\beta}-\beta)\right|\right)\right]+E_{n}\left[I\left(\left|\epsilon_{i}\right|\leq\left|r_{i}\right|\right)\right]\\
 & \lesssim_{P}\bar{E}\left[I\left(\left|\epsilon_{i}\right|\leq\left|d_{i}(\theta-\tilde{\theta})\right|\right)\right]+\bar{E}\left[I\left(\left|\epsilon_{i}\right|\leq\left|x_{i}^{\prime}(\tilde{\beta}-\beta)\right|\right)\right]+\bar{E}\left[\left|\epsilon_{i}\right|\leq\left|r_{i}\right|\right]\\
 & \leq\bar{f}\left\{ \dfrac{1}{n}\sum_{i=1}^{n}\left|d_{i}(\theta-\tilde{\theta})\right|+\dfrac{1}{n}\sum_{i=1}^{n}\left|x_{i}^{\prime}(\tilde{\beta}-\beta)\right|+\dfrac{1}{n}\sum_{i=1}^{n}\left|r_{i}\right|\right\} \\
 & \leq\bar{f}\left\{ \sqrt{\dfrac{1}{n}\sum_{i=1}^{n}\left|d_{i}(\theta-\tilde{\theta})\right|^{2}}+\sqrt{\dfrac{1}{n}\sum_{i=1}^{n}\left|x_{i}^{\prime}(\tilde{\beta}-\beta)\right|^{2}}+\sqrt{\dfrac{1}{n}\sum_{i=1}^{n}\left|r_{i}\right|^{2}}\right\} \\
 & \lesssim_{P}\bar{f}\left\{ \left|\theta-\tilde{\theta}\right|\sqrt{\bar{E}[d_{i}^{2}]}+\sqrt{\bar{E}\left[\left|x_{i}^{\prime}(\tilde{\beta}-\beta)\right|^{2}\right]}+\sqrt{\bar{E}\left[\left|r_{i}\right|^{2}\right]}\right\} \\
 & \lesssim\bar{f}\sqrt{\dfrac{s\ln(pn)}{n}}=\delta_{n}
\end{align*}

\subsection*{A.3 Proof of Theorem 1 (DML 1 case)}

\subsubsection*{A.3.1 Lemma}
\begin{lem}
Under the condition IQR, we have for all $k\in[K]$,
\[
\bar{\sigma}_{n}^{-1}\sqrt{n}(\check{\theta}_{k}-\theta)=\bar{E}_{n}\left[\psi^{2}(w_{i},\theta_{0},\eta_{0})\right]^{-1/2}\dfrac{1}{\sqrt{n}}\sum_{i\in I_{k}}\psi(w_{i},\theta_{0},\eta_{0})\overset{d}{\to}N(0,1)
\]
where $\bar{\sigma}_{n}^{2}=\left(\bar{E}_{n}[t_{i}I(h_{i}>0)f_{i}d_{i}v_{i}]\right)^{-1}\bar{E}_{n}\left[\psi^{2}(w_{i},\theta_{0},\eta_{0})\right]\left(\bar{E}_{n}[t_{i}I(h_{i}>0)f_{i}d_{i}v_{i}]\right)^{-1}$. 
\end{lem}
\begin{proof}
In the proof, I use $(\check{\theta},\hat{h})$ instead of $(\check{\theta}_{k},\hat{h}_{k})$.

\hspace{-5mm}\textbf{Step 1. Normality result}

We have the following identity
\begin{align*}
E_{n}[\psi_{\check{\theta},\hat{\eta}}] & =E_{n}[\psi_{\theta_{0},\eta_{0}}]+E_{n}[\psi_{\check{\theta},\hat{\eta}}-\psi_{\theta_{0},\eta_{0}}]\\
 & =E_{n}[\psi_{\theta_{0},\eta_{0}}]+\underbrace{\bar{E}[\psi_{\check{\theta},\hat{\eta}}]}_{(1)}+\underbrace{n^{-1/2}\mathbb{G}_{n}(\psi_{\check{\theta},\hat{\eta}}-\psi_{\check{\theta},\eta_{0}})}_{(2)}+\underbrace{n^{-1/2}\mathbb{G}_{n}(\psi_{\check{\theta},\eta_{0}}-\psi_{\theta_{0},\eta_{0}})}_{(3)}.
\end{align*}
By the second relation in IQR (3), the left hand side satisfies $\left|E_{n}[\psi_{\check{\theta},\hat{\eta}}]\right|\lesssim\delta_{n}n^{-1/2}$
with probability at least $1-\Delta_{n}$. By step 2, $(1)=-\bar{E}\left[t_{i}I\left(h_{i}>0\right)f_{i}d_{i}v_{i}\right](\hat{\theta}-\theta_{0})+O(\delta_{n}n^{-1/2}+\delta_{n}\left|\hat{\theta}-\theta_{0}\right|)$.
Since $\hat{\eta}=(\hat{g},\hat{v},\hat{\pi})\in\overline{\mathcal{F}}$
with probability at least $1-\Delta_{n}$ by IQR (3), equation (15)
holds. By condition IQR (3) we have with probability at least $1-\Delta_{n}$
that $|\check{\theta}-\theta_{0}|\leq\delta_{n}$. Therefore, we can
apply equation (15) and we have $\left|(2)\right|\lesssim\delta_{n}n^{-1/2}$.
By condition IQR (3) we have with probability at least $1-\Delta_{n}$
that $|\check{\theta}-\theta_{0}|\leq\delta_{n}$. Observe that
\begin{align*}
(\psi_{\theta,\eta_{0}}-\psi_{\theta_{0},\eta_{0}})(y_{i},d_{i},z_{i}) & =t_{i}I\left(h_{i}>0\right)\left(I(y_{i}-d_{i}\theta_{0}-g_{i}\leq0)-I(y_{i}-d_{i}\theta-g_{i}\leq0)\right)v_{i}\\
 & =t_{i}I\left(h_{i}>0\right)\left(I(\epsilon_{i}\leq0)-I(\epsilon_{i}\leq d_{i}(\theta-\theta_{0}))\right)v_{i}
\end{align*}
so that $\left|(\psi_{\theta,\eta_{0}}-\psi_{\theta_{0},\eta_{0}})(y_{i},d_{i},z_{i})\right|\leq I\left\{ \left|\epsilon_{i}\right|\leq\delta_{n}\left|d_{i}\right|\right\} \left|v_{i}\right|$
whenever $|\theta-\theta_{0}|\leq\delta_{n}$. Since the class of
functions 

\hspace{-5mm}$\left\{ (y,d,z)\mapsto(\psi_{\theta,\eta_{0}}-\psi_{\theta_{0},\eta_{0}})(y,d,z):|\theta-\theta_{0}|\leq\delta_{n}\right\} $
is a VC subgraph class, we have
\[
\sup_{|\theta-\theta_{0}|\leq\delta_{n}}\left|\mathbb{G}_{n}(\psi_{\theta,\eta_{0}}-\psi_{\theta_{0},\eta_{0}})\right|\lesssim\left(\bar{E}\left[I\left\{ \left|\epsilon_{i}\right|\leq\delta_{n}\left|d_{i}\right|\right\} v_{i}^{2}\right]\right)^{1/2}\lesssim\delta_{n}^{1/2}.
\]
This implies that $\left|(3)\right|\lesssim\delta_{n}^{1/2}n^{-1/2}$
with probability $1-o(1)$. 

Combining these relations we have
\begin{align*}
\bar{E}\left[t_{i}I\left(h_{i}>0\right)f_{i}d_{i}v_{i}\right](\hat{\theta}-\theta_{0}) & =E_{n}[\psi_{\theta_{0},h_{0}}]+O_{p}(\delta_{n}n^{-1/2})+O_{p}(\delta_{n})\left|\hat{\theta}-\theta_{0}\right|\\
\sqrt{n}(\hat{\theta}-\theta_{0})\left\{ \bar{E}\left[t_{i}I\left(h_{i}>0\right)f_{i}d_{i}v_{i}\right]-O_{p}(\delta_{n})\right\}  & =\sqrt{n}E_{n}[\psi_{\theta_{0},h_{0}}]+O_{p}(\delta_{n}).
\end{align*}
Note that by the Lyapunov CLT
\[
\dfrac{\sqrt{n}E_{n}[\psi_{\theta_{0},h_{0}}]}{\bar{E}[\psi_{\theta_{0},h_{0}}^{2}]^{1/2}}\overset{d}{\to}N(0,1)
\]
\end{proof}
\[
\bar{\sigma}_{n}^{-1}\sqrt{n}(\theta-\theta_{0})=\dfrac{\sqrt{n}E_{n}[\psi_{\theta_{0},h_{0}}]}{\bar{E}[\psi_{\theta_{0},h_{0}}^{2}]^{1/2}}\overset{d}{\to}N(0,1)
\]
where $\bar{\sigma}_{n}^{2}=\bar{E}\left[t_{i}I\left(h_{i}>0\right)f_{i}d_{i}v_{i}\right]^{-1}\bar{E}[\psi_{\theta_{0},h_{0}}^{2}]\bar{E}\left[t_{i}I\left(h_{i}>0\right)f_{i}d_{i}v_{i}\right]^{-1}$
by $\bar{E}\left[t_{i}I\left(h_{i}>0\right)f_{i}d_{i}v_{i}\right]\geq c$
in IQR (1).

\medskip

\hspace{-5mm}\textbf{Step 2}. $(1)=-\bar{E}\left[t_{i}I\left(h_{i}>0\right)f_{i}d_{i}v_{i}\right](\theta-\theta_{0})+O(\delta_{n}\left|\theta-\theta_{0}\right|C+\delta_{n}n^{-1/2})$

\hspace{-5mm}Step 2-1. Bounding $\Gamma(\theta,\hat{\eta})$ for
$\left|\theta-\theta_{0}\right|\leq\delta_{n}$

For any (fixed function) $\hat{\eta}\in\overline{\mathcal{F}}$, we
have
\begin{align*}
\bar{E}[\psi_{\theta,\hat{\eta}}] & =\bar{E}[\psi_{\theta,\eta_{0}}]+\bar{E}[\psi_{\theta,\hat{\eta}}]-\bar{E}[\psi_{\theta,\eta_{0}}]\\
 & =\bar{E}[\psi_{\theta,\eta_{0}}]+\left\{ \bar{E}[\psi_{\theta,\hat{\eta}}]-\bar{E}[\psi_{\theta,\eta_{0}}]-\Gamma_{\eta}(\theta,\eta_{0})(\hat{\eta}-\eta_{0})\right\} +\Gamma_{\eta}(\theta,\eta_{0})(\hat{\eta}-\eta_{0})\\
 & =\underbrace{\bar{E}[\psi_{\theta,\eta_{0}}]}_{(a)}+\underbrace{\left\{ \bar{E}[\psi_{\theta,\hat{\eta}}]-\bar{E}[\psi_{\theta,\eta_{0}}]-\Gamma_{\eta}(\theta,\eta_{0})(\hat{\eta}-\eta_{0})\right\} }_{(b)}+\underbrace{\Gamma_{\eta}(\theta_{0},\eta_{0})[\hat{\eta}-\eta_{0}]}_{(c)}\\
 & \;\;+\underbrace{\left\{ \Gamma_{\eta}(\theta,\eta_{0})[\hat{\eta}-\eta_{0}]-\Gamma_{\eta}(\theta_{0},\eta_{0})[\hat{\eta}-\eta_{0}]\right\} }_{(d)}.
\end{align*}
$(a)$: By Taylor expansion, there is some $\tilde{\theta}\in[\theta_{0},\theta]$
such that
\[
\bar{E}[\psi_{\theta,\eta_{0}}]=\Gamma(\theta_{0},\eta_{0})+\Gamma_{\theta}(\tilde{\theta},\eta_{0})(\theta-\theta_{0})=\left\{ \Gamma_{\theta}(\theta_{0},\eta_{0})+\eta_{n}\right\} (\theta-\theta_{0})
\]
where $\left|\eta_{n}\right|\leq\left|\theta-\theta_{0}\right|C\lesssim\delta_{n}C$
by step 2-2.

\hspace{-5mm}$(b)=\Gamma_{\eta\eta}(\theta,\tilde{\eta})[\hat{\eta}-\eta_{0},\hat{\eta}-\eta_{0}]\lesssim\delta_{n}n^{-1/2}$
by step 2-4.

\hspace{-5mm}$(c)\lesssim\delta_{n}n^{-1/2}$ and $(d)\lesssim\left|\theta-\theta_{0}\right|\delta_{n}$
by step 2-3.

\hspace{-5mm}Combining the arguments, we have
\begin{align*}
\bar{E}[\psi_{\theta,\hat{\eta}}] & =\left\{ \Gamma_{\theta}(\theta_{0},\eta_{0})+O(\delta_{n}C)\right\} (\theta-\theta_{0})+O(\delta_{n}n^{-1/2})+O(\delta_{n}n^{-1/2})+O(\left|\theta-\theta_{0}\right|\delta_{n})\\
 & =\Gamma_{\theta}(\theta_{0},h_{0})(\theta-\theta_{0})+O(\delta_{n}\left|\theta-\theta_{0}\right|C+\delta_{n}n^{-1/2})\\
 & =-\bar{E}\left[t_{i}I\left(h_{i}>0\right)f_{i}d_{i}v_{i}\right](\theta-\theta_{0})+O(\delta_{n}\left|\theta-\theta_{0}\right|C+\delta_{n}n^{-1/2}).
\end{align*}

\hspace{-5mm}Step 2-2. Relations for $\Gamma_{\theta}$ in $(a)$
\begin{align*}
\Gamma_{\theta}(\theta,\tilde{\eta}) & =-\bar{E}\left[t_{i}I\left(\tilde{h}_{i}>0\right)f_{\epsilon_{i}|d_{i},z_{i}}(d_{i}(\theta-\theta_{0})+\tilde{g_{i}}-g_{i})d_{i}\tilde{v}_{i}\right]\\
\Gamma_{\theta}(\theta_{0},\eta_{0}) & =-\bar{E}\left[t_{i}I\left(h_{i}>0\right)f_{i}d_{i}v_{i}\right].
\end{align*}
Moreover, $\Gamma_{\theta}$ also satisfies
\begin{align*}
\left|\Gamma_{\theta}(\theta,\eta_{0})-\Gamma_{\theta}(\theta_{0},\eta_{0})\right| & =\left|\bar{E}\left[t_{i}I\left(h_{i}>0\right)f_{\epsilon_{i}|d_{i},z_{i}}(d_{i}(\theta-\theta_{0}))d_{i}v_{i}\right]-\bar{E}\left[t_{i}I\left(h_{i}>0\right)f_{i}d_{i}v_{i}\right]\right|\\
 & =\left|\bar{E}\left[t_{i}I\left(h_{i}>0\right)\left\{ f_{\epsilon_{i}|d_{i},z_{i}}(d_{i}(\theta-\theta_{0}))-f_{i}\right\} d_{i}v_{i}\right]\right|\\
 & \leq\left|\theta-\theta_{0}\right|\bar{f}^{\prime}\bar{E}\left[d_{i}^{2}v_{i}\right]\\
 & \leq\left|\theta-\theta_{0}\right|\bar{f}^{\prime}\left\{ \bar{E}\left[d_{i}^{4}\right]\bar{E}\left[v_{i}^{2}\right]\right\} ^{1/2}\leq C\left|\theta-\theta_{0}\right|
\end{align*}
since $\bar{E}\left[d_{i}^{4}\right]\vee\bar{E}\left[v_{i}^{2}\right]\leq C$
and $\bar{f}^{\prime}\leq C$ from IQR (1).

\hspace{-5mm}Step 2-3. Relations for $\Gamma_{\eta}$ in $(c),\;(d)$
\begin{align*}
\Gamma_{\eta}(\theta,\tilde{\eta})[\hat{\eta}-\eta_{0}] & =-\bar{E}[t_{i}I\left(\tilde{h}_{i}>0\right)f_{\epsilon_{i}|d_{i},z_{i}}\left(d_{i}(\theta-\theta_{0})+\tilde{g_{i}}-g_{i}\right)\tilde{v}_{i}(\hat{g}_{i}-g_{i})]\\
 & \;\;+\bar{E}\left[\left\{ t_{i}I\left(\tilde{h}_{i}>0\right)\left(\tilde{h}_{i}-I(y_{i}-d_{i}\theta-\tilde{g}_{i}\leq0)\right)+(t_{i}-\tilde{\pi}_{i})I\left(\tilde{h}_{i}>0\right)\dfrac{(1-\tau)}{\tilde{\pi}_{i}}\right\} (\hat{v}_{i}-v_{i})\right]
\end{align*}
Note that when $\Gamma_{\eta}$ is evaluated at $(\theta_{0},\eta_{0})$,
we have with probability $1-\Delta_{n}$,
\begin{align*}
\left|\Gamma_{\eta}(\theta_{0},\eta_{0})[\hat{\eta}-\eta_{0}]\right| & =\left|-\bar{E}[t_{i}I\left(h_{i}>0\right)f_{i}v_{i}(\hat{g}_{i}-g_{i})]\right|\leq\delta_{n}n^{-1/2}
\end{align*}
by equation (10).

The expression for $\Gamma_{\eta}$ also leads to the following bound
\begin{align*}
 & \left|\Gamma_{\eta}(\theta,h_{0})[\hat{\eta}-\eta_{0}]-\Gamma_{\eta}(\theta_{0},\eta_{0})[\hat{\eta}-\eta_{0}]\right|\\
 & \leq\left|\bar{E}[t_{i}I\left(h_{i}>0\right)\left\{ f_{i}-f_{\epsilon_{i}|d_{i},z_{i}}\left(d_{i}(\theta-\theta_{0})\right)\right\} v_{i}(\hat{g}_{i}-g_{i})]\right|+\\
 & \;\;+\left|\bar{E}\left[\left\{ t_{i}I\left(h_{i}>0\right)\left(F_{i}(0)-F_{i}(d_{i}(\theta-\theta_{0}))\right)\right\} (\hat{v}_{i}-v_{i})\right]\right|\\
 & \leq\left|\bar{E}[\left\{ f_{i}-f_{\epsilon_{i}|d_{i},z_{i}}\left(d_{i}(\theta-\theta_{0})\right)\right\} v_{i}(\hat{g}_{i}-g_{i})]\right|+\left|\bar{E}\left[\left(F_{i}(0)-F_{i}(d_{i}(\theta-\theta_{0}))\right)(\hat{v}_{i}-v_{i})\right]\right|\\
 & \leq\bar{E}\left[\left|\theta-\theta_{0}\right|\bar{f}^{\prime}\left|d_{i}v_{i}\right|\left|\hat{g}_{i}-g_{i}\right|\right]+\bar{E}\left[\bar{f}\left|d_{i}(\theta-\theta_{0})\right|\left|\hat{v}_{i}-v_{i}\right|\right]\\
 & \leq\bar{f}^{\prime}\left|\theta-\theta_{0}\right|\left\{ \bar{E}\left[d_{i}^{2}v_{i}^{2}\right]\bar{E}\left[\left|\hat{g}_{i}-g_{i}\right|^{2}\right]\right\} ^{1/2}+\bar{f}\left|\theta-\theta_{0}\right|\left\{ \bar{E}\left[d_{i}^{2}\right]\bar{E}\left[(\hat{v}-v_{i})^{2}\right]\right\} ^{1/2}\\
 & \lesssim_{P}\left|\theta-\theta_{0}\right|\delta_{n}
\end{align*}
by equation (9), (11), and $\bar{E}\left[d_{i}^{2}v_{i}^{2}\right]$,
$\bar{E}\left[d_{i}^{2}\right]$ being bounded.

\hspace{-5mm}Step 2-4. Second derivative (d)

Provided $\hat{\eta}\in\overline{\mathcal{F}}$,
\begin{align*}
\Gamma_{\eta\eta}(\theta,\tilde{\eta})[\hat{\eta}-\eta_{0},\hat{\eta}-\eta_{0}] & =-\bar{E}\left[t_{i}I\left(\tilde{h}_{i}>0\right)f_{\epsilon_{i}|d_{i},z_{i}}^{\prime}\left(d_{i}(\theta-\theta_{0})+\tilde{g}_{i}-g_{i}\right)\tilde{v}_{i}(\hat{g}_{i}-g_{i})^{2}\right]\\
 & \;\;-2\bar{E}\left[t_{i}I\left(\tilde{h}_{i}>0\right)f_{\epsilon_{i}|d_{i},z_{i}}\left(d_{i}(\theta-\theta_{0})+\tilde{g}_{i}-g_{i}\right)(\hat{g}_{i}-g_{i})(\hat{v}_{i}-v_{i})\right]
\end{align*}
\begin{align*}
\left|\partial_{\eta}\partial_{\eta^{\prime}}\bar{E}[\psi(\theta_{0},\tilde{\eta})][\hat{\eta}-\eta_{0},\hat{\eta}-\eta_{0}]\right| & \leq\bar{f^{\prime}}\bar{E}\left[\left|\tilde{v}_{i}\right|(\hat{g}_{i}-g_{i})^{2}\right]\\
 & \;\;+2\bar{f}\bar{E}\left[\left|\hat{v}_{i}-v_{i}\right|\cdot\left|\hat{g}_{i}-g_{i}\right|\right]\\
 & \leq\bar{f^{\prime}}\bar{E}\left[\left|v_{i}\right|(\hat{g}_{i}-g_{i})^{2}\right]+\bar{f^{\prime}}\bar{E}\left[\left|\hat{v}_{i}-v_{i}\right|(\hat{g}_{i}-g_{i})^{2}\right]\\
 & \;\;+2\bar{f}\bar{E}\left[\left|\hat{v}_{i}-v_{i}\right|\cdot\left|\hat{g}_{i}-g_{i}\right|\right]
\end{align*}
since $\left|\tilde{v}_{i}\right|\leq\left|v_{i}\right|+\left|\hat{v}_{i}-v_{i}\right|$.
By equation (12), (13), (14), $\left|\partial_{\eta}\partial_{\eta^{\prime}}\bar{E}[\psi(\theta_{0},\tilde{\eta})][\hat{\eta}-\eta_{0},\hat{\eta}-\eta_{0}]\right|\lesssim\delta_{n}n^{-1/2}$.

\subsubsection*{A.3.2 Main result}

Claim: Let Lemma 2 hold. Then,
\[
\bar{\sigma}_{N}^{-1}\sqrt{N}(\tilde{\theta}-\theta)\overset{d}{\to}N(0,1)
\]
where $\bar{\sigma}_{N}^{2}=\left(\bar{E}_{N}[t_{i}I(h_{i}>0)f_{i}d_{i}v_{i}]\right)^{-1}\bar{E}_{N}\left[\psi^{2}(w_{i},\theta_{0},\eta_{0})\right]\left(\bar{E}_{N}[t_{i}I(h_{i}>0)f_{i}d_{i}v_{i}]\right)^{-1}$. 
\begin{proof}
First, note that$\bar{\sigma}_{N}^{-1}\bar{\sigma}_{n}\to1$ and $\bar{E}_{N}\left[\psi^{2}(w_{i},\theta_{0},\eta_{0})\right]^{1/2}\bar{E}_{n}\left[\psi^{2}(w_{i},\theta_{0},\eta_{0})\right]^{-1/2}\to1$
as $N\to\infty$. Then by the Lyapunov CLT we have{\small{}
\begin{align*}
\bar{\sigma}_{N}^{-1}\sqrt{N}(\tilde{\theta}-\theta) & =\bar{\sigma}_{N}^{-1}\sqrt{N}\left(\dfrac{1}{K}\sum_{k=1}^{K}\check{\theta}_{k}-\theta\right)\\
 & =\bar{\sigma}_{N}^{-1}\sqrt{N}\dfrac{1}{K}\sum_{k=1}^{K}\left(\bar{\sigma}_{n}\bar{E}_{n}\left[\psi^{2}(w_{i},\theta_{0},\eta_{0})\right]^{-1/2}\dfrac{1}{n}\sum_{i\in I_{k}}\psi(\theta_{0},\eta_{0})\right)\\
 & =(\bar{\sigma}_{N}^{-1}\bar{\sigma}_{n})(\bar{E}_{N}\left[\psi^{2}(w_{i},\theta_{0},\eta_{0})\right]^{1/2}\bar{E}_{n}\left[\psi^{2}(w_{i},\theta_{0},\eta_{0})\right]^{-1/2})\cdot\bar{E}_{N}\left[\psi^{2}(w_{i},\theta_{0},\eta_{0})\right]^{-1/2}\dfrac{1}{\sqrt{N}}\sum_{k=1}^{K}\sum_{i\in I_{k}}\psi(w_{i},\theta_{0},\eta_{0})\\
 & =(\bar{\sigma}_{N}^{-1}\bar{\sigma}_{n})(\bar{E}_{N}\left[\psi^{2}(w_{i},\theta_{0},\eta_{0})\right]^{1/2}\bar{E}_{n}\left[\psi^{2}(w_{i},\theta_{0},\eta_{0})\right]^{-1/2})\cdot\left(\bar{E}_{N}\left[\psi^{2}(w_{i},\theta_{0},\eta_{0})\right]^{-1/2}\dfrac{1}{\sqrt{N}}\sum_{i=1}^{N}\psi(w_{i},\theta_{0},\eta_{0})\right)\\
 & \overset{d}{\to}N(0,1).
\end{align*}
}{\small\par}
\end{proof}

\subsection*{A.4 Proof of Theorem 1 (DML 2 case)}

We have the following identity
\begin{align*}
\dfrac{1}{K}\sum_{k=1}^{K}E_{n}[\psi_{\tilde{\theta},\hat{\eta}_{k}}] & =\dfrac{1}{K}\sum_{k=1}^{K}E_{n}[\psi_{\theta_{0},\eta_{0}}]+\dfrac{1}{K}\sum_{k=1}^{K}E_{n}[\psi_{\tilde{\theta},\hat{\eta}_{k}}-\psi_{\theta_{0},\eta_{0}}]\\
 & =\dfrac{1}{K}\sum_{k=1}^{K}E_{n}[\psi_{\theta_{0},\eta_{0}}]+\dfrac{1}{K}\sum_{k=1}^{K}\bar{E}[\psi_{\tilde{\theta},\hat{\eta}_{k}}]+\dfrac{1}{K}\sum_{k=1}^{K}n^{-1/2}\mathbb{G}_{n}(\psi_{\tilde{\theta},\hat{\eta}_{k}}-\psi_{\tilde{\theta},\eta_{0}})\\
 & \;\;+\dfrac{1}{K}\sum_{k=1}^{K}n^{-1/2}\mathbb{G}_{n}(\psi_{\tilde{\theta},\eta_{0}}-\psi_{\theta_{0},\eta_{0}})\\
 & =E_{N}[\psi_{\theta_{0},\eta_{0}}]+\dfrac{1}{K}\sum_{k=1}^{K}\bar{E}[\psi_{\tilde{\theta},\hat{\eta}_{k}}]+\dfrac{1}{K}\sum_{k=1}^{K}n^{-1/2}\mathbb{G}_{n}(\psi_{\tilde{\theta},\hat{\eta}_{k}}-\psi_{\tilde{\theta},\eta_{0}})+N^{-1/2}\mathbb{G}_{N}(\psi_{\tilde{\theta},\eta_{0}}-\psi_{\theta_{0},\eta_{0}}).
\end{align*}
By definition, $\left|\dfrac{1}{K}\sum_{k=1}^{K}E_{n}[\psi_{\tilde{\theta},\hat{\eta}_{k}}]\right|\lesssim\delta_{N}N^{-1/2}$.
By the proof of Lemma 2,
\begin{align*}
\dfrac{1}{K}\sum_{k=1}^{K}n^{-1/2}\mathbb{G}_{n}(\psi_{\tilde{\theta},\hat{\eta}_{k}}-\psi_{\tilde{\theta},\eta_{0}}) & \lesssim\dfrac{1}{K}\sum_{k=1}^{K}\delta_{n}n^{-1/2}=\sqrt{K}\delta_{n}N^{-1/2}=\delta_{N}N^{-1/2}
\end{align*}
\[
\dfrac{1}{K}\sum_{k=1}^{K}\bar{E}[\psi_{\tilde{\theta},\hat{\eta}_{k}}]=-\bar{E}_{N}\left[t_{i}I\left(h_{i}>0\right)f_{i}d_{i}v_{i}\right](\hat{\theta}-\theta_{0})+O_{p}(\delta_{N}N^{-1/2})+O_{p}(\delta_{N}\left|\hat{\theta}-\theta_{0}\right|)
\]
\[
\left|N^{-1/2}\mathbb{G}_{N}(\psi_{\tilde{\theta},\eta_{0}}-\psi_{\theta_{0},\eta_{0}})\right|\lesssim\delta_{N}^{1/2}N^{-1/2}.
\]
Therefore,
\[
\bar{E}_{N}\left[t_{i}I\left(h_{i}>0\right)f_{i}d_{i}v_{i}\right](\hat{\theta}-\theta_{0})=E_{N}[\psi_{\theta_{0},\eta_{0}}]+O_{p}(\delta_{N}N^{-1/2})+O_{p}(\delta_{N})\left|\hat{\theta}-\theta_{0}\right|
\]
and by the Lyapunov CLT
\[
\dfrac{\sqrt{N}E_{N}[\psi_{\theta_{0},h_{0}}]}{\bar{E}[\psi_{\theta_{0},h_{0}}^{2}]^{1/2}}\overset{d}{\to}N(0,1)
\]

\[
\bar{\sigma}_{N}^{-1}\sqrt{N}(\theta-\theta_{0})=\dfrac{\sqrt{N}E_{N}[\psi_{\theta_{0},h_{0}}]}{\bar{E}[\psi_{\theta_{0},h_{0}}^{2}]^{1/2}}\overset{d}{\to}N(0,1).
\]

\subsection*{A.5 Proof of Theorem 2}

\subsubsection*{A.5.1 Subsample result}

\textbf{Claim 1.}
\begin{align*}
\dfrac{1}{n}\sum_{i\in I_{k}}t_{i}I(\hat{h}_{k,i}>0)\hat{f}_{k,i}d_{i}\hat{v}_{k,i}-\dfrac{1}{n}\sum_{i\in I_{k}}E\left[t_{i}I(h_{i}>0)f_{i}d_{i}v_{i}\right] & \to0
\end{align*}

\begin{proof}
In the proof, I use $E_{n}$, $\hat{\eta}$, $\left\Vert \cdot\right\Vert _{2,n}$
instead of $E_{n,k}$, $\hat{\eta}_{k}$, $\left\Vert \cdot\right\Vert _{2,n,k}$.
\begin{align*}
 & \left|E_{n}[t_{i}I(\hat{h}_{i}>0)\hat{f}_{i}d_{i}\hat{v}_{i}]-\bar{E}_{n}[t_{i}I\left(h_{i}>0\right)f_{i}d_{i}v_{i}]\right|\\
 & \leq\left|E_{n}[t_{i}I(\hat{h}_{i}>0)\hat{f}_{i}d_{i}\hat{v}_{i}]-E_{n}[t_{i}I\left(h_{i}>0\right)f_{i}d_{i}v_{i}]\right|+\left|E_{n}[t_{i}I\left(h_{i}>0\right)f_{i}d_{i}v_{i}]-\bar{E}_{n}[t_{i}I\left(h_{i}>0\right)f_{i}d_{i}v_{i}]\right|\\
 & \leq\left|E_{n}[t_{i}I(\hat{h}_{i}>0)\hat{f}_{i}d_{i}\hat{v}_{i}]-E_{n}[t_{i}I\left(\hat{h}_{i}>0\right)f_{i}d_{i}v_{i}]\right|+\left|E_{n}[t_{i}I\left(\hat{h}_{i}>0\right)f_{i}d_{i}v_{i}]-E_{n}[t_{i}I\left(h_{i}>0\right)f_{i}d_{i}v_{i}]\right|\\
 & \;\;+\left|E_{n}[t_{i}I\left(h_{i}>0\right)f_{i}d_{i}v_{i}]-\bar{E}_{n}[t_{i}I\left(h_{i}>0\right)f_{i}d_{i}v_{i}]\right|\\
 & \leq\left|E_{n}[d_{i}\left(\hat{f}_{i}\hat{v}_{i}-f_{i}v_{i}\right)]\right|+\left|E_{n}[\left\{ I\left(\hat{h}_{i}>0\right)-I\left(h_{i}>0\right)\right\} f_{i}d_{i}v_{i}]\right|+\left|E_{n}[t_{i}I\left(h_{i}>0\right)f_{i}d_{i}v_{i}]-\bar{E}_{n}[t_{i}I\left(h_{i}>0\right)f_{i}d_{i}v_{i}]\right|\\
 & \leq\left|E_{n}[(\hat{f}_{i}-f_{i})d_{i}\hat{v}_{i}]\right|+\left|E_{n}[f_{i}d_{i}(\hat{v}_{i}-v_{i})]\right|+\left|E_{n}[\left\{ I\left(\hat{h}_{i}>0\right)-I\left(h_{i}>0\right)\right\} f_{i}d_{i}v_{i}]\right|\\
 & \;\;+\left|E_{n}[t_{i}I\left(h_{i}>0\right)f_{i}d_{i}v_{i}]-\bar{E}_{n}[t_{i}I\left(h_{i}>0\right)f_{i}d_{i}v_{i}]\right|\\
 & \leq\left|E_{n}[(\hat{f}_{i}-f_{i})d_{i}(\hat{v}_{i}-v_{i})]\right|+\left|E_{n}[(\hat{f}_{i}-f_{i})d_{i}v_{i}]\right|+\left|E_{n}[f_{i}d_{i}(\hat{v}_{i}-v_{i})]\right|\\
 & \;\;+\left|E_{n}[\left\{ I\left(\hat{h}_{i}>0\right)-I\left(h_{i}>0\right)\right\} f_{i}d_{i}v_{i}]\right|+\left|E_{n}[t_{i}I\left(h_{i}>0\right)f_{i}d_{i}v_{i}]-\bar{E}_{n}[t_{i}I\left(h_{i}>0\right)f_{i}d_{i}v_{i}]\right|\\
 & \lesssim_{P}\left\Vert (\hat{f}_{i}-f_{i})d_{i}\right\Vert _{2,n}\left\Vert \hat{v}_{i}-v_{i}\right\Vert _{2,n}+\left\Vert \hat{f}_{i}-f_{i}\right\Vert _{2,n}\left\Vert d_{i}v_{i}\right\Vert _{2,n}+\left\Vert f_{i}d_{i}\right\Vert _{2,n}\left\Vert \hat{v}_{k,i}-v_{i}\right\Vert _{2,n}\\
 & \;\;+\left\Vert f_{i}d_{i}v_{i}\right\Vert _{2,n}\left\Vert I\left(\hat{h}_{i}>0\right)-I\left(h_{i}>0\right)\right\Vert _{2,n}+\left|E_{n}[t_{i}I\left(h_{i}>0\right)f_{i}d_{i}v_{i}]-\bar{E}_{n}[t_{i}I\left(h_{i}>0\right)f_{i}d_{i}v_{i}]\right|
\end{align*}
Since $f_{i},\hat{f}_{i}\leq C$, $\bar{E}[d_{i}^{4}]\leq C$, $\bar{E}[v_{i}^{4}]\leq C$
by Condition IQR (1), $\left\Vert \hat{v}_{i}-v_{i}\right\Vert _{2,n}\lesssim_{P}\delta_{n}$,

\hspace{-5mm}$\left\Vert I\left(\hat{h}_{k,i}>0\right)-I\left(h_{i}>0\right)\right\Vert _{2,n}\lesssim_{P}\delta_{n}$,
and Law of large numbers,
\[
\left|E_{n}[t_{i}I(\hat{h}_{i}>0)\hat{f}_{i}d_{i}\hat{v}_{i}]-\bar{E}_{n}[t_{i}I\left(h_{i}>0\right)f_{i}d_{i}v_{i}]\right|\lesssim_{P}\delta_{n}.
\]
\end{proof}
\textbf{Claim 2.}

\[
\dfrac{1}{n}\sum_{i\in I_{k}}\psi^{2}(w_{i},\tilde{\theta},\hat{\eta}_{k})-\dfrac{1}{n}\sum_{i\in I_{k}}E\left[\psi^{2}(w_{i},\theta_{0},\eta_{0})\right]\to0
\]

\begin{proof}
In the proof, I use $E_{n}$, $\hat{\eta}$, $\left\Vert \cdot\right\Vert _{2,n}$
instead of $E_{n,k}$, $\hat{\eta}_{k}$, $\left\Vert \cdot\right\Vert _{2,n,k}$.
Since $E_{n}\left[\psi^{2}(w_{i},\theta_{0},\eta_{0})\right]-\bar{E}\left[\psi^{2}(w_{i},\theta_{0},\eta_{0})\right]\to_{P_{n}}0$
by Law of large numbers, it suffices to show $E_{n}\left[\psi^{2}(w_{i},\tilde{\theta},\hat{\eta})\right]-E_{n}\left[\psi^{2}(w_{i},\theta_{0},\eta_{0})\right]\to_{P_{n}}0$.
Since
\begin{align*}
 & E_{n}\left[\psi^{2}(w_{i},\tilde{\theta},\hat{\eta}_{k})\right]-E_{n}\left[\psi^{2}(w_{i},\theta_{0},\eta_{0})\right]=\\
 & \left(\sqrt{E_{n}\left[\psi^{2}(w_{i},\tilde{\theta},\hat{\eta}_{k})\right]}-\sqrt{E_{n}\left[\psi^{2}(w_{i},\theta_{0},\eta_{0})\right]}\right)\cdot\left(\sqrt{E_{n}\left[\psi^{2}(w_{i},\tilde{\theta},\hat{\eta}_{k})\right]}+\sqrt{E_{n}\left[\psi^{2}(w_{i},\theta_{0},\eta_{0})\right]}\right)
\end{align*}
it suffices to show that $\left|\left\Vert \psi(w_{i},\tilde{\theta},\hat{\eta})\right\Vert _{2,n}-\left\Vert \psi(w_{i},\theta_{0},\eta_{0})\right\Vert _{2,n}\right|\to_{P_{n}}0$.
\begin{align*}
 & \left|\left\Vert \psi(w_{i},\tilde{\theta},\hat{\eta})\right\Vert _{2,n}-\left\Vert \psi(w_{i},\theta_{0},\eta_{0})\right\Vert _{2,n}\right|\\
 & \leq\left\Vert \psi(w_{i},\tilde{\theta},\hat{\eta})-\psi(w_{i},\theta_{0},\eta_{0})\right\Vert _{2,n}\\
 & \leq\underbrace{\left\Vert t_{i}I\left(\hat{h}_{i}>0\right)\left(\hat{h}_{i}-I(y_{i}-d_{i}\tilde{\theta}-\hat{g}_{i}\leq0)\right)\hat{v}_{i}-t_{i}I\left(h_{i}>0\right)\left(h_{i}-I(y_{i}-d_{i}\theta_{0}-g_{i}\leq0)\right)v_{i}\right\Vert _{2,n}}_{A}\\
 & \;\;+\underbrace{\left\Vert I\left(\hat{h}_{i}>0\right)(t_{i}-\hat{\pi}_{i})\dfrac{(1-\tau)}{\hat{\pi}_{i}}\hat{v}_{i}-I\left(h_{i}>0\right)(t_{i}-\pi_{i})\dfrac{(1-\tau)}{\pi_{i}}v_{i}\right\Vert _{2,n}}_{B}
\end{align*}
\begin{align*}
A & \leq\left\Vert t_{i}I\left(\hat{h}_{i}>0\right)\left(\hat{h}_{i}-I(y_{i}-d_{i}\tilde{\theta}-\hat{g}_{i}\leq0)\right)(\hat{v}_{i}-v_{i})\right\Vert _{2,n}\\
 & \;\;+\left\Vert \left\{ t_{i}I\left(\hat{h}_{i}>0\right)\left(\hat{h}_{i}-I(y_{i}-d_{i}\tilde{\theta}-\hat{g}_{i}\leq0)\right)-t_{i}I\left(h_{i}>0\right)\left(h_{i}-I(y_{i}-d_{i}\theta-g_{i}\leq0)\right)\right\} v_{i}\right\Vert _{2,n}\\
 & \leq\left\Vert t_{i}I\left(\hat{h}_{i}>0\right)\left(\hat{h}_{i}-I(y_{i}-d_{i}\tilde{\theta}-\hat{g}_{i}\leq0)\right)(\hat{v}_{i}-v_{i})\right\Vert _{2,n}\\
 & \;\;+\left\Vert \left\{ t_{i}I\left(\hat{h}_{i}>0\right)\left(\hat{h}_{i}-I(y_{i}-d_{i}\tilde{\theta}-\hat{g}_{i}\leq0)\right)-t_{i}I\left(\hat{h}_{i}>0\right)\left(h_{i}-I(y_{i}-d_{i}\theta-g_{i}\leq0)\right)\right\} v_{i}\right\Vert _{2,n}\\
 & \;\;+\left\Vert \left\{ t_{i}I\left(\hat{h}_{i}>0\right)\left(h_{i}-I(y_{i}-d_{i}\theta-g_{i}\leq0)\right)-t_{i}I\left(h_{i}>0\right)\left(h_{i}-I(y_{i}-d_{i}\theta-g_{i}\leq0)\right)\right\} v_{i}\right\Vert _{2,n}\\
 & \leq\left\Vert t_{i}I\left(\hat{h}_{i}>0\right)\left(\hat{h}_{i}-I(y_{i}-d_{i}\tilde{\theta}-\hat{g}_{i}\leq0)\right)(\hat{v}_{i}-v_{i})\right\Vert _{2,n}\\
 & \;\;+\left\Vert \left\{ t_{i}I\left(\hat{h}_{i}>0\right)\left(\hat{h}_{i}-h_{i}-I(y_{i}-d_{i}\tilde{\theta}-\hat{g}_{i}\leq0)+I(y_{i}-d_{i}\theta-g_{i}\leq0)\right)\right\} v_{i}\right\Vert _{2,n}\\
 & \;\;+\left\Vert \left\{ I\left(\hat{h}_{i}>0\right)-I\left(h_{i}>0\right)\right\} t_{i}\left(h_{i}-I(y_{i}-d_{i}\theta-g_{i}\leq0)\right)v_{i}\right\Vert _{2,n}\\
 & \leq\left\Vert \left(\hat{h}_{i}-I(y_{i}-d_{i}\tilde{\theta}-\hat{g}_{i}\leq0)\right)(\hat{v}_{i}-v_{i})\right\Vert _{2,n}\\
 & \;\;+\left\Vert v_{i}\left(\hat{h}_{i}-h_{i}\right)\right\Vert _{2,n}\\
 & \;\;+\left\Vert v_{i}\left(I(y_{i}-d_{i}\tilde{\theta}-\hat{g}_{i}\leq0)-I(y_{i}-d_{i}\theta-g_{i}\leq0)\right)\right\Vert _{2,n}\\
 & \;\;+\left\Vert \left(h_{i}-I(y_{i}-d_{i}\theta-g_{i}\leq0)\right)v_{i}\left\{ I\left(\hat{h}_{i}>0\right)-I\left(h_{i}>0\right)\right\} \right\Vert _{2,n}
\end{align*}
\begin{align*}
\left\Vert \left(\hat{h}_{k,i}-I(y_{i}-d_{i}\tilde{\theta}-\hat{g}_{k,i}\leq0)\right)(\hat{v}_{k,i}-v_{i})\right\Vert _{2,n} & \lesssim_{P}\left\Vert \hat{v}_{k,i}-v_{i}\right\Vert _{2,n}\lesssim_{P}\delta_{n}\\
\left\Vert v_{i}\left(\hat{h}_{k,i}-h_{i}\right)\right\Vert _{2,n} & \lesssim_{P}\left\Vert \hat{\pi}_{k,i}-\pi_{i}\right\Vert _{2,n}\lesssim_{P}\delta_{n}
\end{align*}
\begin{align*}
\left\Vert v_{i}\left(I(y_{i}-d_{i}\tilde{\theta}-\hat{g}_{k,i}\leq0)-I(y_{i}-d_{i}\theta-g_{i}\leq0)\right)\right\Vert _{2,n} & \leq\left\Vert v_{i}I\left(\left|\epsilon_{i}\right|\leq\left|d_{i}(\theta-\tilde{\theta})+g_{i}-\tilde{g}_{i}\right|\right)\right\Vert _{2,n}\\
 & \leq\left\{ \left\Vert v_{i}^{2}\right\Vert _{2,n}\left\Vert I\left(\left|\epsilon_{i}\right|\leq\left|d_{i}(\theta-\tilde{\theta})+g_{i}-\tilde{g}_{i}\right|\right)\right\Vert _{2,n}\right\} ^{1/2}\lesssim_{P}\delta_{n}
\end{align*}
\[
\left\Vert \left(h_{i}-I(y_{i}-d_{i}\theta-g_{i}\leq0)\right)v_{i}\left\{ I\left(\hat{h}_{k,i}>0\right)-I\left(h_{i}>0\right)\right\} \right\Vert _{2,n}\lesssim_{P}\left\Vert I\left(\hat{h}_{k,i}>0\right)-I\left(h_{i}>0\right)\right\Vert _{2,n,k}\lesssim_{P}\delta_{n}
\]
Therefore, $A\lesssim_{P}\delta_{n}$.
\begin{align*}
B & =\left\Vert (1-\tau)t_{i}\cdot I\left(\hat{h}_{i}>0\right)\dfrac{\hat{v}_{i}}{\hat{\pi}_{i}}-(1-\tau)\cdot I\left(\hat{h}_{i}>0\right)\hat{v}_{i}-(1-\tau)t_{i}\cdot I\left(h_{i}>0\right)\dfrac{v_{i}}{\pi_{i}}+(1-\tau)\cdot I\left(h_{i}>0\right)v_{i}\right\Vert _{2,n}\\
 & \lesssim\left\Vert I\left(\hat{h}_{i}>0\right)\dfrac{\hat{v}_{i}}{\hat{\pi}_{i}}-I\left(h_{i}>0\right)\dfrac{v_{i}}{\pi_{i}}\right\Vert _{2,n}+\left\Vert I\left(\hat{h}_{i}>0\right)\hat{v}_{i}-I\left(h_{i}>0\right)v_{i}\right\Vert _{2,n}\\
 & \leq\left\Vert I\left(\hat{h}_{i}>0\right)\dfrac{\hat{v}_{i}}{\hat{\pi}_{i}}-I\left(\hat{h}_{i}>0\right)\dfrac{v_{i}}{\hat{\pi}_{i}}\right\Vert _{2,n}+\left\Vert I\left(\hat{h}_{i}>0\right)\dfrac{v_{i}}{\hat{\pi}_{i}}-I\left(\hat{h}_{i}>0\right)\dfrac{v_{i}}{\pi_{i}}\right\Vert _{2,n}+\left\Vert I\left(\hat{h}_{i}>0\right)\dfrac{v_{i}}{\pi_{i}}-I\left(h_{i}>0\right)\dfrac{v_{i}}{\pi_{i}}\right\Vert _{2,n}\\
 & +\left\Vert I\left(\hat{h}_{i}>0\right)\hat{v}_{i}-I\left(\hat{h}_{i}>0\right)v_{i}\right\Vert _{2,n}+\left\Vert I\left(\hat{h}_{i}>0\right)v_{i}-I\left(h_{i}>0\right)v_{i}\right\Vert _{2,n}\\
 & =\left\Vert I\left(\hat{h}_{i}>0\right)\dfrac{\hat{v}_{i}-v_{i}}{\hat{\pi}_{i}}\right\Vert _{2,n}+\left\Vert I\left(\hat{h}_{i}>0\right)v_{i}\dfrac{1}{\hat{\pi}_{i}\pi_{i}}(\hat{\pi}_{i}-\pi_{i})\right\Vert _{2,n}+\left\Vert \dfrac{v_{i}}{\pi_{i}}\left\{ I\left(\hat{h}_{i}>0\right)-I\left(h_{i}>0\right)\right\} \right\Vert _{2,n}\\
 & +\left\Vert I\left(\hat{h}_{i}>0\right)\left(\hat{v}_{i}-v_{i}\right)\right\Vert _{2,n}+\left\Vert v_{i}\left\{ I\left(\hat{h}_{i}>0\right)-I\left(h_{i}>0\right)\right\} \right\Vert _{2,n}
\end{align*}
Since $\hat{\pi}_{i}$ and $\pi_{i}$ are bounded away from zero,
$\bar{E}[v_{i}^{2}]\leq C$, $\left\Vert \hat{v}_{i}-v_{i}\right\Vert _{2,n}\lesssim_{P}\delta_{n}$,
$\left\Vert I\left(\hat{h}_{k,i}>0\right)-I\left(h_{i}>0\right)\right\Vert _{2,n}\lesssim_{P}\delta_{n}$,
$\left\Vert \hat{\pi}_{i}-\pi_{i}\right\Vert _{2,n}\lesssim_{P}\delta_{n}$,
$B\lesssim_{P}\delta_{n}$.
\end{proof}

\subsubsection*{A.5.2 Main result}

Under Claim 1 and Claim 2, the variance estimator $\hat{\sigma}_{N}^{2}$
is consistent where
\begin{align*}
\hat{\sigma}_{N}^{2} & =\left(\dfrac{1}{K}\sum_{k=1}^{K}E_{n,k}[t_{i}I(\hat{h}_{k,i}>0)\hat{f}_{k,i}d_{i}\hat{v}_{k,i}]\right)^{-1}\dfrac{1}{K}\sum_{k=1}^{K}E_{n,k}\left[\psi^{2}(w_{i},\tilde{\theta},\hat{\eta}_{k})\right]\left(\dfrac{1}{K}\sum_{k=1}^{K}E_{n,k}[t_{i}I(\hat{h}_{k,i}>0)\hat{f}_{k,i}d_{i}\hat{v}_{k,i}]\right)^{-1}.
\end{align*}

\begin{proof}
First, by Claim 1, for all $k\in[K]$
\begin{align*}
\left|\dfrac{1}{n}\sum_{i\in I_{k}}t_{i}I(\hat{h}_{k,i}>0)\hat{f}_{k,i}d_{i}\hat{v}_{k,i}-\dfrac{1}{n}\sum_{i\in I_{k}}E\left[t_{i}I(h_{i}>0)f_{i}d_{i}v_{i}\right]\right| & \to0.
\end{align*}
Then,
\begin{align*}
 & \dfrac{1}{N}\sum_{i=1}^{N}E[t_{i}I(h_{i}>0)f_{i}d_{i}v_{i}]-\dfrac{1}{K}\sum_{k=1}^{K}E_{n,k}[t_{i}I(\hat{h}_{k,i}>0)\hat{f}_{k,i}d_{i}\hat{v}_{k,i}]\\
 & =\dfrac{1}{K}\sum_{k=1}^{K}\left\{ \dfrac{1}{n}\sum_{i\in I_{k}}E\left[t_{i}I(h_{i}>0)f_{i}d_{i}v_{i}\right]-\dfrac{1}{n}\sum_{i\in I_{k}}t_{i}I(\hat{h}_{k,i}>0)\hat{f}_{k,i}d_{i}\hat{v}_{k,i}\right\} \\
 & \to0.
\end{align*}
Next, by Claim 2, for all $k\in[K]$
\begin{align*}
\left|\dfrac{1}{n}\sum_{i\in I_{k}}\psi^{2}(w_{i},\tilde{\theta},\hat{\eta}_{k})-\dfrac{1}{n}\sum_{i\in I_{k}}E\left[\psi^{2}(w_{i},\theta_{0},\eta_{0})\right]\right| & \to0.
\end{align*}
Then,
\begin{align*}
\dfrac{1}{N}\sum_{i=1}^{N}E\left[\psi^{2}(w_{i},\theta_{0},\eta_{0})\right]-\dfrac{1}{K}\sum_{k=1}^{K}E_{n,k}[\psi^{2}(w_{i},\tilde{\theta},\hat{\eta}_{0,k})] & =\dfrac{1}{K}\sum_{k=1}^{K}\left\{ \dfrac{1}{n}\sum_{i\in I_{k}}E\left[\psi^{2}(w_{i},\theta_{0},\eta_{0})\right]-\dfrac{1}{n}\sum_{i\in I_{k}}\psi^{2}(w_{i},\tilde{\theta},\hat{\eta}_{k})\right\} \\
 & \to0.
\end{align*}
\end{proof}
\newpage

\bibliographystyle{ecta}
\bibliography{reference}

\end{document}